\documentclass[a4paper,UKenglish,cleveref, autoref,numberwithinsect]{lipics-v2019}
\usepackage[utf8]{inputenc}
\usepackage{amssymb}
\usepackage{amsthm}
\usepackage{mathtools}
\usepackage{float}
\usepackage{stmaryrd}
\usepackage{subcaption}
\usepackage{mathtools}
\usepackage[dvipsnames]{xcolor}
\usepackage{tikz}
\usetikzlibrary{calc}
\usetikzlibrary{shapes.geometric}
\usetikzlibrary{shapes.misc}
\usetikzlibrary{shapes.symbols}
\usetikzlibrary{positioning}
\usetikzlibrary{decorations.markings}
\usepackage{stackengine}
\setstackgap{S}{1pt}

\nolinenumbers
\hideLIPIcs

\newcommand\SmallMatrix[1]{{%
		\footnotesize\arraycolsep=0.1pt\renewcommand*{\arraystretch}{1}\ensuremath{\begin{bmatrix}#1\end{bmatrix}}}}
\newcommand{\bvdots}{ \tikz[baseline, every node/.style={inner sep=0}]{ \node at (0,0){.}; \node at (0,-6pt){.}; \node at (0,6pt){.}; } }
\newcommand{\interp}[1]{\left\llbracket#1\right\rrbracket}
\newcommand{\ket}[1]{|#1\rangle}
\newcommand{\bra}[1]{\langle#1|}
\newcommand{\szx}{{SZX}}

\makeatletter
\renewcommand*{\eqref}[1]{%
	\hyperref[{#1}]{\textup{\tagform@{\ref*{#1}}}}%
}
\makeatother

\usepackage{tikzit}

\tikzset{strike through/.style={
		postaction=decorate,
		decoration={
			markings,
			mark=at position 0.5 with {
				\draw[-] (-2.5pt,-2.5pt) -- (2.5pt, 2.5pt);
			}
		}
	}
}

\tikzset{strike bend/.style={
		postaction=decorate,
		decoration={
			markings,
			mark=at position 0.5 with {
				\draw[-] (0pt,-3pt) -- (0pt, 3pt);
			}
		}
	}
}
\definecolor{green}{HTML}{ccffcc}
\definecolor{red}{HTML}{ff8888}
\definecolor{zx_grey}{RGB}{211,211,211}
\tikzstyle{gn}=[fill=green, draw=black, shape=circle,  tikzit category=ZX, tikzit fill=green, tikzit draw=black, tikzit shape=circle, inner sep=2pt]
\tikzstyle{rn}=[fill=red, draw=black, shape=circle, tikzit fill=red, tikzit draw=black, tikzit category=ZX, tikzit shape=circle, inner sep=2pt]
\tikzstyle{divide}=[regular polygon, regular polygon sides=3, shape border rotate=90, draw=black,fill=zx_grey, inner sep=1.6pt, tikzit category=scal, rounded corners=0.8mm]
\tikzstyle{black}=[fill=black, draw=black, shape=circle, tikzit fill=black, tikzit draw=black, tikzit shape=circle, tikzit category=IH, inner sep=2pt]
\tikzstyle{gather}=[fill=zx_grey, draw=black, tikzit category=scal, rounded corners=0.8mm, regular polygon, regular polygon sides=3, shape border rotate=-90, inner sep=1.6pt]
\tikzstyle{ggen}=[fill=white, draw=black, shape=rectangle, rounded corners=2mm,  line width=1pt, tikzit draw=red, tikzit category=scal]
\tikzstyle{white}=[fill=white, draw=black, shape=circle, inner sep=2pt, tikzit category=IH]
\tikzstyle{mbox}=[fill=white, draw=black, rounded rectangle, rounded rectangle west arc=none, tikzit category=scal, tikzit shape=rectangle]
\tikzstyle{A}=[fill=white, shape=circle, tikzit category=scal, inner sep=1pt]
\tikzstyle{ggreen}=[fill=green, draw=black, shape=circle, tikzit category=SZX, tikzit fill=green, tikzit draw=black, line width=1pt, inner sep=2pt]
\tikzstyle{gred}=[fill=red, draw=black, shape=circle, rounded corners=2mm,  tikzit category=SZX, inner sep=2pt, tikzit fill=red, line width=1pt]
\tikzstyle{ghad}=[fill=yellow, draw=black, shape=rectangle, tikzit category=SZX, tikzit shape=rectangle, tikzit fill=yellow, inner sep=2pt, line width=1pt]
\tikzstyle{boxm}=[fill=white, draw=black, rounded rectangle, tikzit category=scal, tikzit shape=rectangle, rounded rectangle east arc=none]
\tikzstyle{box}=[fill=white, draw=black, shape=rectangle]
\tikzstyle{had}=[fill=yellow, draw=black, shape=rectangle, tikzit category=ZX, tikzit fill=yellow, tikzit draw=black, inner sep=2pt]
\tikzstyle{gwhite}=[fill=white, draw=black, shape=circle, tikzit fill=white, tikzit shape=circle, line width=1 pt, inner sep=2 pt, tikzit draw=red]
\tikzstyle{gblack}=[fill=black, draw=black, shape=circle, tikzit fill=black, tikzit shape=circle, line width=1 pt, inner sep=2 pt, tikzit draw=red]
\tikzstyle{antipode}=[fill=red, draw=black, shape=rectangle, tikzit fill=red, tikzit draw=black, tikzit shape=rectangle, inner sep=2pt]
\tikzstyle{diamond}=[fill=white, draw=black, shape=diamond, inner sep=2pt]
\tikzstyle{mongr}=[fill=green, draw=green, shape=circle, inner sep=2pt]
\tikzstyle{monbl}=[fill=blue, draw=blue, shape=circle, inner sep=2pt]
\tikzstyle{bg}=[inner sep=0.7mm, minimum width=0pt, minimum height=0pt, fill=green, draw=white, very thick, shape=circle]
\tikzstyle{br}=[inner sep=0.7mm, minimum width=0pt, minimum height=0pt, fill=red, draw=white, very thick, shape=circle]
\tikzstyle{rmat}=[draw, signal, fill=zx_grey, signal to=east, signal from=west, inner sep=1pt, minimum height=6pt]
\tikzstyle{lmat}=[draw, signal, fill=zx_grey, signal to=west, signal from=east, inner sep=1pt, minimum height=6pt]
\tikzstyle{umat}=[draw, signal, fill=zx_grey, signal to=north, signal from=south, inner sep=1pt, minimum width=6pt]
\tikzstyle{dmat}=[draw, signal, fill=zx_grey, signal to=south, signal from=north, inner sep=1pt, minimum width=6pt]
\tikzstyle{arrow}=[->]
\tikzstyle{very thick}=[-, line width=1pt, tikzit draw=red]
\tikzstyle{pointille}=[dashed, -]
\tikzstyle{red}=[-, draw=red]
\tikzstyle{blue}=[-, draw=blue]
\tikzstyle{green}=[-, draw=green]
\tikzstyle{arrow}=[->]
\tikzstyle{strike}=[-, tikzit draw={rgb,255: red,191; green,0; blue,64}, strike through]
\tikzstyle{strike'}=[-, tikzit draw=cyan, strike bend]

\tikzstyle{box}=[shape=rectangle, text height=1.5ex, text depth=0.25ex, yshift=0.5mm, fill=white, draw=black, minimum height=5mm, yshift=-0.5mm, minimum width=5mm, font={\small}]
\tikzstyle{Z dot}=[inner sep=0mm, minimum size=2mm, shape=circle, draw=black, fill={rgb,255: red,160; green,255; blue,160}]
\tikzstyle{gdot}=[minimum size=3mm, font={\scriptsize\boldmath}, shape=rectangle, rounded corners=1.3mm, inner sep=1mm, outer sep=-1.8mm, scale=0.8, tikzit shape=circle, draw=black, fill=green, tikzit draw=blue]
\tikzstyle{X dot}=[Z dot, shape=circle, draw=black, fill={rgb,255: red,220; green,0; blue,0}]
\tikzstyle{rdot}=[minimum size=3mm, font={\scriptsize\boldmath}, shape=rectangle, rounded corners=1.3mm, inner sep=1mm, outer sep=-1.8mm, scale=0.8, tikzit shape=circle, draw=black, fill=red, tikzit draw=blue]
\tikzstyle{grdot}=[minimum size=3mm, font={\scriptsize\boldmath}, shape=rectangle, rounded corners=1.3mm, inner sep=1mm,
 line width=1pt, outer sep=-1.5mm, scale=0.8, tikzit shape=circle, draw=black, fill=red, tikzit draw=blue]
 \tikzstyle{ggdot}=[minimum size=3mm, font={\scriptsize\boldmath}, shape=rectangle,  line width=1pt, rounded corners=1.3mm, inner sep=1mm, outer sep=-1.5mm, scale=0.8, tikzit shape=circle, draw=black, fill=green, tikzit draw=blue]

\title{\szx-calculus:~Scalable~Graphical Quantum Reasoning}

\titlerunning{\szx-calculus}

\author{Titouan Carette}{CNRS, LORIA, Inria Mocqua, Universit\'e de Lorraine, F 54000 Nancy, France}{titouan.carette@loria.fr}{https://orcid.org/0000-0002-1618-4081}{}
\author{Dominic Horsman}{LIG, Universit\'e Grenoble Alpes, France}{dominic.horsman@univ-grenoble-alpes.fr}{}{}
\author{Simon Perdrix}{ CNRS, LORIA, Inria Mocqua, Universit\'e de Lorraine, F 54000 Nancy, France}{simon.perdrix@loria.fr}{https://orcid.org/0000-0002-1808-2409}{}

\authorrunning{T.Carette, D. Horsman, and S.Perdrix}
\Copyright{T.Carette, D. Horsman, and S.Perdrix}

\ccsdesc[100]{Theory of computation~Quantum Computation theory}

\keywords{Quantum computing, categorical quantum mechanics, completeness, scalability}

\category{}

\relatedversion{}
\supplement{}

\funding{This work is funded by ANR-17-CE25-0009 SoftQPro, ANR-17-CE24-0035 VanQuTe, PIA-GDN/Quantex, and  LUE / UOQ.}

\begin{document}

\maketitle
\begin{abstract}
We introduce the Scalable ZX-calculus (\szx-calculus for short), a formal and compact graphical language for the design and verification of quantum computations. The \szx-calculus is an extension of the ZX-calculus, a powerful framework that captures graphically the fundamental properties of quantum mechanics through its complete set of rewrite rules. The ZX-calculus is, however, a low level language, with each wire representing a single qubit. This limits its ability to handle large and elaborate quantum evolutions. We extend the ZX-calculus to registers of qubits and allow compact representation of sub-diagrams via binary matrices. 
We show soundness and completeness of the \szx-calculus and provide two examples of applications, for graph states and error correcting codes. \end{abstract}

\section{Introduction}
The ZX-calculus is an intuitive and powerful graphical language for quantum computing, introduced by Coecke and Duncan \cite{coecke2011interacting}. Quantum processes can be represented by ZX-diagrams, which can be seen intuitively as a generalisation of quantum circuits. The language is also equipped with a set of rewrite rules which preserves the represented quantum evolution. Unlike quantum circuits, the  ZX-calculus has been proved to be complete for various universal fragments of pure quantum mechanics \cite{jeandel2018complete,HNW,ZXNormalForm,vilmart2018near}, and also mixed states quantum mechanics \cite{carette2019completeness}. Completeness means that any equality can be derived in this language: if two diagrams represent the same quantum process then they can be transformed one into the other using the rewriting rules of the language. Completeness opens avenues for various applications of the ZX-calculus in quantum information processing, including circuit optimisation \cite{duncan2019graph,kissinger2019cnot} -- which out-performs all other technics for T-count reductions \cite{kissinger2019reducing} -- error correcting codes \cite{duncan2014steane,gidney2018efficient,chancellor2016graphical}, lattice surgery \cite{de2017zx}, measurement-based quantum computing \cite{DP-2010,duncan2012graphical,Kissinger-MBQC} \textit{etc}. Automated tools for quantum reasoning, e.g. Quantomatic \cite{kissinger2015quantomatic} and PyZX \cite{pyzx}, are also based on the ZX-calculus. The ZX-calculus is also used as intermediate representation in a commercial quantum compiler \cite{cowtan2019qubit}.

The cornerstone of the ZX-calculus is that fundamental properties of quantum mechanics can be captured graphically. The language remains, however, relatively low level: each wire represents a single qubit, a feature that limits the design of larger-scale and more complex quantum procedures. We address in this paper the problem of scalability of the ZX-calculus. In \cite{chancellor2016graphical}, the authors -- including one of the present paper --  demonstrated that the ZX-calculus can be used in practice to design and verify quantum error correcting codes. They  introduced various shortcuts to deal with the scalability of the language: mainly the use of thick wires to represent registers of qubits and matrices to represent sub-diagrams, and hence reason about families of diagrams in a compact way. However, the approach lacked a general theory and fundamental properties like soundness and completeness.

\noindent {\bf Contributions.}  We introduce the Scalable ZX-calculus, \szx~calculus for short, to provide theoretical foundations to this approach. We extend the ZX-calculus to deal with registers of qubits by introducing some new generators and rewrite rules. We show soundness -- i.e. the new generators can be used in a consistent way --  as well as completeness of the \szx-calculus. 
A  simple but key ingredient is the introduction of two generators, not present in \cite{chancellor2016graphical}, for dividing and gathering registers of qubits. A wire representing a register of $(n{+}m)$-qubits can be divided into two wires representing respectively $n$ and $m$ qubits. Similarly two registers can be gathered into a single larger one. We also extend the generators of the ZX-calculus so that they can act not only on a single qubit but on a register of qubits. The \szx-calculus is then constructed as a combination of the ZX-calculus and the sub-language made of the divider and the gatherer, by adding the necessary rewrite rules describing how these two sub-languages interact. We show that the \szx-calculus is universal, sound,  and complete, providing an intuitive and formal language to represent quantum operations on an arbitrarily large finite number of qubits. The use of the divider and the gatherer allows one to derive inductive (graphical) proofs.

Furthermore, the \szx-calculus provides the fundamental structures -- namely the (co)comu\-tative Hopf algebras -- to develop a graphical theory of binary matrices, following work on graphical linear algebra \cite{bonchi2017interacting}. As a consequence, we introduce an additional generator parametrized by a binary matrix together with four simple rewrite rules. Note that, while matrices were also used in \cite{chancellor2016graphical}, we introduce here a more elementary generator acting on a single register (1 input/1 output) rather than two registers (2 inputs/2 outputs). We prove  completeness of the \szx-calculus augmented with these matrices. 
The use of matrices allows a compact representation where subdiagrams can be replaced by  matrices. Moreover, basic matrix arithmetic can be done graphically. It makes the \szx-calculus with matrices a powerful tool for formal and compact quantum reasoning.  

In section \ref{sec:applications}, we show the \szx-calculus in action. The main application of the \szx-calculus we consider in this paper is the graph state formalism  \cite{hein2006entanglement}. We show how graph states can be represented using \szx-diagrams and how some fundamental properties like fixpoint properties, local complementation, and pivoting can be derived in the calculus. We also consider error correcting code examples in order to show that the techniques for the design and verification of codes developed in \cite{chancellor2016graphical} can be performed smoothly in the \szx-calculus.

\noindent{\bf Related works.} Scalability is crucial in the development of the  ZX-calculus and more generally for graphical languages. We review here some contributions in this domain that we briefly compare to our approach. 

The !-boxes formalism \cite{kissinger2015first} is a meta language for graphical languages, which  has been extensively used in the development of the automated tool Quantomatic. A !-box is a  region (subdiagram) of a diagram which can be  discarded or duplicated. There is also a first order logic handling families of equations between concrete (i.e.~!-box free) diagrams.  
In contrast, the scalable ZX is not a meta-language but an actual graphical language equipped with an equational theory (namely a coloured PROP). There is no obvious way to compare these two approaches (even in terms of expressive power). 

Monoidal multiplexing \cite{chantawibul2018monoidal} corresponds to two categorical constructions which allow representing $n$ diagrams in parallel. Roughly speaking, one of the two constructions would be equivalent to the use of big wires for the subclass of \szx-diagrams which are matrix, divider and gatherer-free. It is worth noticing that, to our knowledge, monoidal multiplexing has never been combined with the matrix approach, even though both were developed in the same line of research on graphical linear algebra.

Recently, Miatto \cite{miatto2019graphical} has independently introduced a graphical calculus involving matrices, and the equivalent of green spiders, dividers and gatherers. This graphical calculus has been developed in the context of the tensor networks, and the author mainly shows that 6 kinds of matrix products can be represented graphically.  We note that the represented matrices do not coincide with the ones we are axiomatising: the matrices represented in Miatto's language correspond to $\mathbb C^{2^{m}\times 2^n}$ matrices whereas ours are in $\mathbb F_2^{m\times n}$, hence the equations differ. It is however worth noting that equation Fig.6 in \cite{miatto2019graphical} essentially corresponds to the equation governing the interaction between green spiders and the divider given in section \ref{dist}.

\section{Background: the ZX-calculus}

A ZX-diagram $D:k\to \ell$ with $k$ inputs and $\ell$ outputs is generated by: $\forall n,m\in \mathbb N$,  $\forall \alpha \in \mathbb R$,\\
 \renewcommand{\arraystretch}{1.5}\arraycolsep=2pt
\centerline{$\begin{array}{rclrclrclrcl}
		\quad\input{gspider.tikz} &:& n\to m~~~~&\qquad\begin{tikzpicture}
	\begin{pgfonlayer}{nodelayer}
		\node [style=had] (0) at (1, 0) {};
		\node [style=none] (1) at (0, 0) {};
		\node [style=none] (2) at (2, 0) {};
	\end{pgfonlayer}
	\begin{pgfonlayer}{edgelayer}
		\draw (1.center) to (0);
		\draw (0) to (2.center);
	\end{pgfonlayer}
\end{tikzpicture}
 &:& 1\to 1~~~~&\qquad\input{cup.tikz}&:& 0\to 2~~~~&\qquad\input{swap.tikz}&:&2\to 2~~~~ \\[0.12cm]
		\input{rspider.tikz} &:& n\to m&\begin{tikzpicture}
	\begin{pgfonlayer}{nodelayer}
		\node [style=none] (13) at (2, 0) {};
		\node [style=none] (14) at (0, 0) {};
	\end{pgfonlayer}
	\begin{pgfonlayer}{edgelayer}
		\draw (14.center) to (13.center);
	\end{pgfonlayer}
\end{tikzpicture}
&:&1\to 1&\input{cap.tikz}&:&2\to 0 &\input{bone1.tikz}&:&0\to0 \\[0.12cm]
	\end{array}$}
and the two compositions: for any ZX-diagrams $D_0 :a \to b$, $D_1 :b\to c$, and $D_2:c\to d$:\vspace{0.1cm}

\centerline{
${\input{D1.tikz}}\circ {\input{D0.tikz}}={\input{compD.tikz}}\qquad$ and $\qquad  {\input{D0.tikz}}\otimes {\input{D2.tikz}}={\input{tensorD.tikz}}$\vspace{0.1cm}}

For any $n,m$, $ZX[n,m]$ is the set of all ZX-diagrams of type $n\to m$. 
The ZX-diagrams are representing quantum processes: for any ZX-diagram $D:n\to m$ its interpretation $\interp D\in \mathcal{M}_{2^m \times 2^n}(\mathbb{C})$ is inductively defined as: $\interp{D_1\circ D_0} = \interp {D_1} \circ \interp {D_0}$, $\interp{D_0\otimes D_2} = \interp {D_0} \otimes \interp {D_2}$, and 
\[\begin{array}{rclrcl}
		\interp{\input{gspider.tikz}}&\!\!\coloneqq \!\!&\ket{0^m}\!\bra{0^n}+e^{i\alpha}\ket{1^m}\!\bra{1^n}&\interp{\begin{tikzpicture}
	\begin{pgfonlayer}{nodelayer}
		\node [style=none] (13) at (2, 0) {};
		\node [style=none] (14) at (0, 0) {};
	\end{pgfonlayer}
	\begin{pgfonlayer}{edgelayer}
		\draw (14.center) to (13.center);
	\end{pgfonlayer}
\end{tikzpicture}
}&\!\!\coloneqq \!\!&\ket{0}\!\bra{0}+\ket{1}\!\bra{1}\\[0.2cm]
		\interp{\input{rspider.tikz}}&\!\!\coloneqq \!\!& \ket{+^m}\!\bra{+^n}+e^{i\alpha}\ket{-^m}\!\bra{-^n}\quad&\quad\interp{\begin{tikzpicture}
	\begin{pgfonlayer}{nodelayer}
		\node [style=had] (0) at (1, 0) {};
		\node [style=none] (1) at (0, 0) {};
		\node [style=none] (2) at (2, 0) {};
	\end{pgfonlayer}
	\begin{pgfonlayer}{edgelayer}
		\draw (1.center) to (0);
		\draw (0) to (2.center);
	\end{pgfonlayer}
\end{tikzpicture}
}&\!\!\coloneqq \!\!&\ket{+}\!\bra{0}+\ket{-}\!\bra{1} \\[0.2cm]
		\interp{\input{swap.tikz}}&\!\!\coloneqq \!\!&\ket{00}\!\bra{00}+\ket{01}\!\bra{10}+\ket{10}\!\bra{01}+\ket{11}\!\bra{11}&\interp{\input{bone1.tikz}}&\!\!\coloneqq \!\!&1\\[0.2cm]
		\interp{\input{cup.tikz}}&\!\!\coloneqq \!\!&\ket{00}+\ket{11}&\interp{\input{cap.tikz}}&\!\!\coloneqq \!\!&\bra{00}+\bra{11}
	\end{array}\]
	Where $\ket{0}\!\!\coloneqq\!\!{1\choose 0}$, $\ket{1}\!\!\coloneqq\!\!{0\choose 1}$, $\ket{+}\!\!\coloneqq\!\!\frac{\ket{0}+\ket{1}}{\sqrt{2}}$, $\ket{-}\!\!\coloneqq\!\!\frac{\ket{0}-\ket{1}}{\sqrt{2}}$, $\ket{a^{k+1}}\!\!\coloneqq\!\!\ket{a}{\otimes} \ket{a^{k}}$, $\ket{a^0}\!\!\coloneqq\!\! 1$,  and $\bra{a}\!\!\coloneqq\!\! \ket{a}^\dagger$, moreover $n$ and $m$ are respectively the number of inputs and outputs of the spiders.

When equal to zero, the angle of the green or red spider is omitted: \vspace{0.2cm}

\centerline{$\input{gspider-0.tikz}\coloneqq\input{gspider-zero.tikz}\qquad$ and $\qquad \input{rspider-0.tikz}\coloneqq\input{rspider-zero.tikz}$\vspace{0.2cm}}

ZX-diagrams are \emph{universal} for pure qubit quantum mechanics: $\forall n,m\in \mathbb N$, and $\forall M\in  \mathcal{M}_{2^n \times 2^m}(\mathbb{C})$, there exists a ZX-diagram $D:n\to m$ such that $\interp D = M$. 

ZX-diagrams also come with a set of graphical rewrite rules, or axioms, which allows one to transform a diagram preserving its interpretation. %First, diagrams can be deformed at will 
%Computation in the ZX-calculus are carried using graphical rewriting rules. 
Some of them are gathered under the \textit{Only Topology Matters} paradigm. When using these we label the equality by top\label{top}. Two diagrams that can be transformed into each other by moving around the wires are equal. This can be derived from the following rules:

\begin{center}
	\begin{tabular}{ccccc}
		$\input{iswap0.tikz}=\input{iswap1.tikz}~~$&$~\input{scup0.tikz}=\input{cup.tikz}~~$&$~\input{scap0.tikz}=\input{cap.tikz}~~$&$~\input{snake0.tikz}=\begin{tikzpicture}
	\begin{pgfonlayer}{nodelayer}
		\node [style=none] (16) at (1, 0) {};
		\node [style=none] (17) at (0, 0) {};
	\end{pgfonlayer}
	\begin{pgfonlayer}{edgelayer}
		\draw (16.center) to (17.center);
	\end{pgfonlayer}
\end{tikzpicture}
=\input{snake2.tikz}~~$&$~\scalebox{0.9}{\input{sdiagp0.tikz}}\!\!=\!\!\scalebox{0.9}{\input{sdiagp1.tikz}}~$\\
	\end{tabular}	
\end{center}

The last set of rule expresses the naturality of the swap; in other words, that all the generators can be passed through wires.

The legs of the spiders of ZX-calculus can be exchanged and bent. This implies that diagrams are essentially graphs with inputs and outputs.

\begin{center}
	\begin{tabular}{cc}
		$\input{sgspider0.tikz}=\input{sgspider1.tikz}\quad$&$\quad\input{bspider0.tikz}=\input{bspider1.tikz}$\\
	\end{tabular}
\end{center}

Finally,  the rules that are not purely topological are given in Figure \ref{fig:ZXaxioms}.

\begin{figure}[!h]
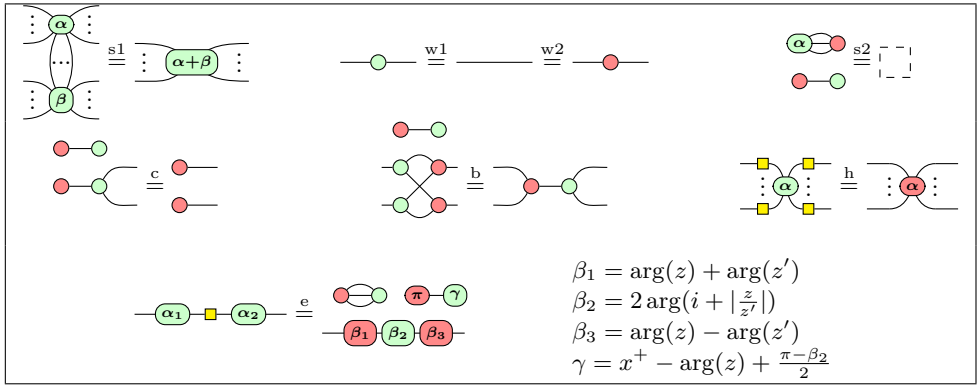

	\begin{tabular}{|ccccc|}\hline
		$\input{spider0.tikz}\stackrel{\textup{s1}}{=}\label{s1}\input{spider1.tikz}$&$\quad$&$\input{wire0.tikz}\stackrel{\textup{w1}}{=}\label{w1}\label{wire}\input{wire1.tikz}\stackrel{\textup{w2}}{=}\label{w2}\input{wire2.tikz}$&$\quad$&$\input{bone0.tikz}\stackrel{\textup{s2}}{=}\label{s2}\input{bone1.tikz}$\\[0.5cm]
		$\input{rcopy0.tikz}\stackrel{\textup{c}}{=}\label{c}\input{rcopy1.tikz}$&$\quad$&$\input{bigebre0.tikz}\stackrel{\textup{b}}{=}\label{b}\input{bigebre1.tikz}$&$\quad$&$\input{hadp0.tikz}\stackrel{\textup{h}}{=}\label{h}\input{hadp1.tikz}$\\[0.5cm]
		\multicolumn{5}{|c|}{\footnotesize \begin{tabular}{ccc}
				$\input{euler0.tikz}\stackrel{\textup{e}}{=}\label{e}\input{euler1.tikz}$&$\quad$&\begin{tabular}{l}
					$\beta_{1}=\arg(z)+\arg(z')$\\[-0.2cm]
					$\beta_{2}=2\arg(i+|\frac{z}{z'}|)$\\[-0.2cm]
					$\beta_{3}=\arg(z)-\arg(z')$\\[-0.2cm]
					$\gamma=x^{+}-\arg(z)+\frac{\pi-\beta_2}{2}$
				\end{tabular}
			\end{tabular}
		}\\\hline
	\end{tabular}
	\caption{{\bf Axioms of the ZX-calculus}, \label{fig:ZXaxioms}where $x^{+}\coloneqq\frac{\alpha_1 + \alpha_2}{2}$, $x^{-}\coloneqq x^{+}- \alpha_2$, $z\coloneqq-\sin(x^{+})+ i\cos(x^{-})$, $z'\coloneqq\cos(x^{+})-i\sin(x^{-})$ and $z'=0 \Rightarrow \beta_2 =0$.
In the upper left rule, there must be at least one wire between the spiders annotated by $\alpha$ and $\beta$. The colour-swapped version of those rules also holds. A label is given to each axiom, above the equals sign, for later reference. 
}
\end{figure}

We write $ZX\vdash D=D'$ when $D$ can be transformed into $D'$ using the rules of the ZX-calculus. The rules of the ZX-calculus are sound: for any ZX-diagrams $D,D'$, $ZX\vdash D=D' \Rightarrow \interp{D}=\interp{D'}$ i.e., the rules preserve the interpretation of the language. The language  is also complete:   for any ZX-diagrams $D,D'$,   $\interp{D}=\interp{Dp}\Rightarrow ZX\vdash D=D'$ i.e., whenever two diagrams represent the same quantum evolution, we can transform one into the other using the rules of the language \cite{vilmart2018near}.

\section{The scalable ZX-calculus}

In the ZX-calculus, each wire represents a single qubit. Therefore, a system acting on $n$ qubits will be represented by an $n$-input diagram. This quickly leads to intractable diagrams when it comes to big systems.
The extension to the \szx-calculus presented here provides a more compact notation.

\subsection{Divide and gather, a calculus for big wires}

The input (resp. output) type of a ZX-diagram is its number of input wires, and hence number of input qubits. 
 In the \szx-calculus, wires represent registers of qubits. A wire of type $1_{n}$ represents a register of $n$ qubits. A type of the \szx-calculus is then a formal sum of the form $\sum_{i}{1}_{n_i}$, the empty sum being denoted by $0$. In other words, the set of types of \szx-calculus is the free monoid over $\mathbb{N}^*$, the set of positive integers. We denote it $\langle\mathbb{N}^* \rangle$. Graphically, we represent the wire of type $1_{n}$ by an  bold font wires labelled by $n$\footnote{On a blackboard the bold font might be advantageously replaced by struck-out wires.},  a label that is omitted when it is not ambiguous. A normal font wire always denotes a single qubit register of type $1_1$. 
By convention the sum of $m$ wires of type $1_n$ is denoted $m_n$ with $0_n=m_0=0$. $n_1$ is simply written $n$. Given a type $a=\sum_{i}{1}_{n_i}$, its size is defined as $S(a)\coloneqq\sum_{i}{n_i}$.

Big wires can be divided into smaller ones and, conversely, can be gathered to form bigger ones.  
For any $n\in \mathbb N$, we introduce two new generators: the \emph{divider} and \emph{gatherer} of size $n$. They are depicted as follows: 

\vspace{-0.2cm}
\begin{center}
\begin{tabular}{ccc}
		\input{divider.tikz}& $\quad$&\input{gatherer.tikz}
	\end{tabular}
	\end{center}
\vspace{-0.2cm}

We take the convention that the divider and the gatherer of size $0$ are the identity. We define a fragment of the \szx, the wire calculus $\mathbb{W}$.

\begin{definition}[$\mathbb{W}$-calculus] 
	The $\mathbb{W}$-calculus is defined as the graphical language generated by identity wires, the dividers, and the gatherers of any size, and satisfying the elimination rule  
		$\input{elim0-s.tikz}\stackrel{\textup{E}}{=}\label{E}\input{elim1-s.tikz}$ and  the expansion rule $\input{exp0.tikz}\stackrel{\textup{P}}{=}\label{P}\begin{tikzpicture}[scale=0.8]
	\begin{pgfonlayer}{nodelayer}
		\node [style=none] (9) at (0, 0) {};
		\node [style=none] (11) at (2, 0) {};
	\end{pgfonlayer}
	\begin{pgfonlayer}{edgelayer}
		\draw [style=very thick] (9.center) to (11.center);
	\end{pgfonlayer}
\end{tikzpicture}
$. 
	
\end{definition}

The roles of the dividers and gatherers in the equations are perfectly symmetric, so each time something is shown for dividers it also holds for gatherers by symmetry.

We now show a coherence theorem for scalable calculi: the rewiring theorem. It states that two diagrams of the $\mathbb{W}$-calculus with the same type are equal.

\begin{theorem}\label{thm:rewire}
	Let $\omega\in \mathbb{W}[a,b]$ and $\omega'\in \mathbb{W}[c,d]$: $~\mathbb{W}\vdash \omega=\omega' ~~\Leftrightarrow~~ a=c ~and~ b=d$
\end{theorem}

\noindent The proof is in the appendix at page \pageref{pr:rewire}.

This theorem has strong consequences. We can define generalized dividers able to divide any wire of size $1_{a+b}$ into a wire of size $1_{a}$ and a wire of size $1_b$.

\vspace{-0.2cm}
\begin{center}
\input{gdiv0.tikz}
\end{center}
\vspace{-0.2cm}

Those generalized dividers have a unique possible interpretation as diagrams of $\mathbb{W}$-calculus given by their types, and we know exactly the equations they verify: all the well typed ones. In particular, an associativity-like law holds for generalized wires allowing us to define $n$-ary generalized dividers.

\vspace{-0.2cm}
\begin{center}
$\input{asdiv0.tikz}\quad\coloneqq\quad\input{asdiv1.tikz}\quad=\quad\input{asdiv2.tikz}$
\end{center}
\vspace{-0.2cm}

Each time we use the property that any well typed equation in $\mathbb{W}$ is true, we will label the equality by \textup{R}\label{R}.

\subsection{The \szx-diagrams}

We now fuse the $\mathbb{W}$-calculus and the ZX-calculus into one language: the full \szx-calculus.

The generators of  \szx-diagrams are: $\forall n,m\in \mathbb N^*,\forall k,\ell\in \mathbb N$,  $\forall \alpha \in \mathbb R^n$,\\
 \renewcommand{\arraystretch}{1.5}\arraycolsep=2pt
\centerline{\!\!\!\!\!\!\!$\begin{array}{rclrclrclrclrcl}
		\quad\input{bgspider.tikz} &:& k_n{\to} \ell_n~~~&\begin{tikzpicture}
	\begin{pgfonlayer}{nodelayer}
		\node [style=ghad] (0) at (0.75, 0) {};
		\node [style=none] (1) at (0, 0) {};
		\node [style=none] (2) at (1.5, 0) {};
	\end{pgfonlayer}
	\begin{pgfonlayer}{edgelayer}
		\draw [style=very thick] (1.center) to (0);
		\draw [style=very thick] (0) to (2.center);
	\end{pgfonlayer}
\end{tikzpicture}
 &:& 1_n{\to} 1_n~~~~~&\input{bcup.tikz}&:& 0{\to} 2_n  ~~~& \input{divider-0.tikz}&:&1_{n+1}{\to} 1{+}1_{n} ~~~&\begin{tikzpicture}[scale=0.8]
	\begin{pgfonlayer}{nodelayer}
		\node [style=none] (9) at (0, 0) {};
		\node [style=none] (11) at (1, 0) {};
	\end{pgfonlayer}
	\begin{pgfonlayer}{edgelayer}
		\draw [style=very thick] (9.center) to (11.center);
	\end{pgfonlayer}
\end{tikzpicture}
 &: &1_{n}{\to} 1_n \\[0.12cm]%\hline
		\input{brspider.tikz} &:&  k_n{\to} \ell_n&\input{bswap.tikz}&:&1_n{+}1_m{\to} 1_m{+}1_n~~~&\input{bcap.tikz}&:&2_n{\to} 0 &\input{gatherer-0.tikz} &:&1{+}1_{n}{\to} 1_{n+1} & \input{bone1.tikz}&:&0{\to}0  \\[0.12cm]%\hline
	\end{array}$}
\szx-generators can be combined using the usual sequential and spacial compositions to form \szx-diagrams. 
Note that for $n=m=1$ we recover all the generators of the ZX-calculus. We denote them, as in the ZX-calculus, using thin wires e.g. $\begin{tikzpicture}
	\begin{pgfonlayer}{nodelayer}
		\node [style=had] (0) at (1, 0) {};
		\node [style=none] (1) at (0, 0) {};
		\node [style=none] (2) at (2, 0) {};
	\end{pgfonlayer}
	\begin{pgfonlayer}{edgelayer}
		\draw (1.center) to (0);
		\draw (0) to (2.center);
	\end{pgfonlayer}
\end{tikzpicture}
$ for  $\begin{tikzpicture}
	\begin{pgfonlayer}{nodelayer}
		\node [style=ghad] (0) at (0.75, 0) {};
		\node [style=none] (1) at (0, 0) {};
		\node [style=none] (2) at (1.5, 0) {};
	\end{pgfonlayer}
	\begin{pgfonlayer}{edgelayer}
		\draw [style=very thick] (1.center) to (0);
		\draw [style=very thick] (0) to (2.center);
	\end{pgfonlayer}
\end{tikzpicture}
 : 1_1\to 1_1$. 
Any big wire can  be labelled by its size $\input{bid-tag.tikz} : 1_n\to 1_n$ to avoid ambiguity. Such labels will be  used mainly for scalars i.e.~diagrams with no input/output. 
Each green or red spider is parametrised by a vector $\alpha \in \mathbb R^n$ of angles. With slight  abuse of notation we use a single angle $\alpha_0\in \mathbb R$ to denote the vector $(\alpha_0, \ldots , \alpha_0)\in \mathbb R^n$  when the spider has at least one leg ($k+\ell>0$) so that this leg can be labelled by $n$ to avoid a potential ambiguity. Like in the ZX-calculus, the angle $\alpha_0$ is omitted when $\alpha_0=0$.

The interpretation of ZX-diagrams is extended to \szx-diagrams as follows: 
for any \szx-diagram $D:a\to b$,  its interpretation $\interp{D}_s$ is a triplet $(M,a,b)$ where $M\in \mathcal{M}_{2^{S(b)} \times 2^{S(a)}}(\mathbb{C})$. $\interp{D}_s$ is inductively defined as: $\interp{D_1\circ D_0}_s = (M_1\circ M_0, a,c)$,  $\interp{D_0\otimes D_2}_s = (M_0\otimes M_2, a+c,b+d)$ where $\interp{D_0}_s=(M_0,a,b)$, $\interp{D_1}_s=(M_1,b,c)$, and $\interp{D_2}_s=(M_2,c,d)$. Moreover: \[\begin{array}{rclrcl}
		\interp{\begin{tikzpicture}
	\begin{pgfonlayer}{nodelayer}
		\node [style=ghad] (0) at (0.75, 0) {};
		\node [style=none] (1) at (0, 0) {};
		\node [style=none] (2) at (1.5, 0) {};
	\end{pgfonlayer}
	\begin{pgfonlayer}{edgelayer}
		\draw [style=very thick] (1.center) to (0);
		\draw [style=very thick] (0) to (2.center);
	\end{pgfonlayer}
\end{tikzpicture}
}_s&\!\!\coloneqq \!\!&(\frac{1}{\sqrt 2^n}\sum\limits_{x,y\in \{0,1\}^n} (-1)^{x\bullet y}\ket{y}\!\bra{x},1_n,1_n) &\interp{\input{divider-0.tikz}}_s&\!\!\coloneqq \!\!&(id_{n+1},1_{n+1},1{+}1_{n})\\[0.2cm]
		\interp{\input{bgspider.tikz}}_s&\!\!\coloneqq \!\!&(\sum\limits_{x\in \{0,1\}^n} e^{i x\bullet \alpha}\ket{x^k}\!\bra{x^\ell},k_n,\ell_n)&\interp{\input{gatherer-0.tikz}}_s&\!\!\coloneqq \!\!&(id_{n+1}, 1{+}1_{n},1_{n+1})\\[0.2cm]
		\interp{\input{brspider.tikz}}_s&\!\!\coloneqq \!\!&{\interp{\begin{tikzpicture}
	\begin{pgfonlayer}{nodelayer}
		\node [style=ghad] (0) at (0.75, 0) {};
		\node [style=none] (1) at (0, 0) {};
		\node [style=none] (2) at (1.5, 0) {};
	\end{pgfonlayer}
	\begin{pgfonlayer}{edgelayer}
		\draw [style=very thick] (1.center) to (0);
		\draw [style=very thick] (0) to (2.center);
	\end{pgfonlayer}
\end{tikzpicture}
}_s}^{\otimes \ell} \circ \interp{\input{bgspider.tikz}}_s \circ  {\interp{\begin{tikzpicture}
	\begin{pgfonlayer}{nodelayer}
		\node [style=ghad] (0) at (0.75, 0) {};
		\node [style=none] (1) at (0, 0) {};
		\node [style=none] (2) at (1.5, 0) {};
	\end{pgfonlayer}
	\begin{pgfonlayer}{edgelayer}
		\draw [style=very thick] (1.center) to (0);
		\draw [style=very thick] (0) to (2.center);
	\end{pgfonlayer}
\end{tikzpicture}
}_s}^{\otimes k}   \quad&\interp{\input{bcup.tikz}}_s&\!\!\coloneqq \!\!&(\sum\limits_{x\in \{0,1\}^n} \ket {xx}, 0,2_n)\\[0.2cm]
		\interp{\input{bswap.tikz}}_s&\!\!\coloneqq \!\!&(\sum\limits_{x\in \{0,1\}^n,y\in \{0,1\}^m}\ket {yx}\!\bra{xy},1_n{+}1_m, 1_m{+}1_n)&\quad\interp{\input{bcap.tikz}}_s&\!\!\coloneqq \!\!&(\sum\limits_{x\in \{0,1\}^n} \bra {xx}, 2_n,0)\\[0.2cm]
			\interp{\begin{tikzpicture}[scale=0.8]
	\begin{pgfonlayer}{nodelayer}
		\node [style=none] (9) at (0, 0) {};
		\node [style=none] (11) at (1, 0) {};
	\end{pgfonlayer}
	\begin{pgfonlayer}{edgelayer}
		\draw [style=very thick] (9.center) to (11.center);
	\end{pgfonlayer}
\end{tikzpicture}
}_s&\!\!\coloneqq \!\!&(id_n,1_n,1_n)&\interp{\input{bone1.tikz}}_s&\!\!\coloneqq \!\!&1
	\end{array}\]
	Where $\forall u,v\in \mathbb R^m$, $u\bullet v = \sum_{i=1}^m u_iv_i$, $M^{\otimes 0} = 1$, and $M^{\otimes k+1} = M\otimes M^{\otimes k}$. 
 
\begin{theorem}[Universality]\label{thm:univ}
\szx-diagrams are universal for pure qubit quantum mechanics: $\forall a,b \in \langle \mathbb N^*\rangle, \forall M\in   \mathcal{M}_{2^{S(b)} \times 2^{S(a)}}(\mathbb{C})$, $\exists D:a\to b$ such that $\interp{D}_s= (M,a,b)$. 
\end{theorem}
\noindent The proof is in the appendix at page \pageref{pr:univ}.

\subsection{The calculus}

The \szx-calculus is based on distribution rules that allow dividers and gatherers to go through the big generators. For this to work we need first to ensure that the swap behaves naturally with respect to dividers and gatherers. This is given by the following two rules:
\vspace{-0.2cm}
\begin{center}
	\begin{tabular}{ccc}
		$\input{dswap0.tikz}~=~\input{dswap1.tikz}\qquad $&$\qquad\input{dswapp0.tikz}~=~\input{dswapp1.tikz}$
	\end{tabular}
\end{center}
\vspace{-0.2cm}
Then the rules governing the interaction between dividers, gatherers and the so-called  \emph{cups} and \emph{caps} are:
\vspace{-0.2cm}
\begin{center}
	\begin{tabular}{cc}
		$\input{icup0-2.tikz}~\stackrel{\textup{U}}{=}\label{U}~\input{icup1-2.tikz}\qquad$&$\qquad\input{icap0-2.tikz}~\stackrel{\textup{A}}{=}\label{A}~\input{icap1-2.tikz}\qquad$
	\end{tabular}
\end{center}
\vspace{-0.2cm}

We put labels over the equals signs to allow subsequent reference to the rules.
These rules are sufficient to fully describe possible interactions between wires of any size, gatherers and dividers.
 It remains to specify how dividers and gatherers interact with big generators:
\vspace{-0.2cm}
\begin{center}
	\begin{tabular}{ccc}\label{dist}
		$\input{gspiderd0.tikz}~\stackrel{\textup{Z}}{=}\label{Z}~\input{gspiderd1.tikz}$&$\input{spiderd0.tikz}~\stackrel{\textup{X}}{=}\label{X}~\input{spiderd1.tikz}$&$\input{hadd0.tikz}~\stackrel{\textup{W}}{=}\label{W}~\input{hadd1.tikz}$
	\end{tabular}
\end{center}
\vspace{-0.2cm}
Where $\alpha{::}\beta$ means that we append the phase $\alpha\in \mathbb R$ to the (generalized) phase $\beta\in \mathbb R^n$.

This completes the set of rules of the \szx-calculus. 
	Note that all rules agree with the interpretation, ensuring soundness of the \szx-calculus.

We see that any big generator $s_n$ is in fact just $n$ copies of the corresponding size one generator $s$ acting in parallel. That is, a parallel composition but with a particular permutation of the inputs and outputs. Such constructions are called multiplexed diagrams in \cite{chantawibul2018monoidal}. Multiplexed diagrams are shown to satisfy the same equations as size $1$ diagrams. The following lemma states the same results for big generators:

\begin{lemma}\label{lm:srule}
For any rule of the ZX-calculus, and any $n\in \mathbb N^*$, the equation obtained by replacing each generator by its big version of size $n$ 
 is provable in the \szx-calculus.
\end{lemma}
\noindent The proof is in the appendix at page \pageref{pr:srule}.

We can go even further than Lemma \ref{lm:srule}. In fact, the \szx-calculus is complete:

\begin{theorem}\label{thm:sem}
	$\forall a,b \in \langle \mathbb N^*\rangle, \forall D,D'\in \mathcal{S}ZX[a,b],~  \interp{D}_s=\interp{Dp}_s\Rightarrow \szx\vdash D=D'$.
\end{theorem}
\noindent The proof is in the appendix at page \pageref{pr:sem}.

Theorem \ref{thm:sem}  has interesting graphical consequences, ensuring that the \textit{Only Topology Matters} paradigm applies to the \szx-calculus. In particular, swaps of any size behave naturally with respect to any diagram:

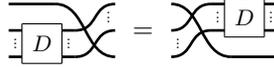
\vspace{-0.2cm}
\begin{center}
	$\input{sdiag0.tikz}~=~\input{sdiag1.tikz}$
\end{center}
\vspace{-0.2cm}

This suggests a more compact presentation close to the one of the $ZX$-calculus, given in the next subsection. 

\subsection{Compact axiomatisation}

Assuming that \textit{Only Topology Matters}, the  \szx-calculus enjoys a more compact axiomatisation: 

\begin{figure}[H]
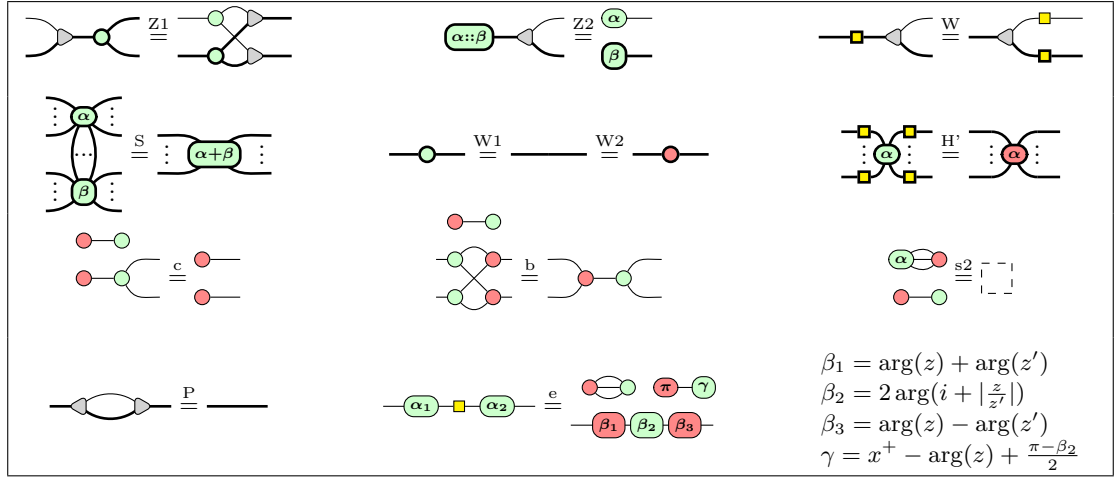

	\begin{tabular}{|ccccc|}\hline
		$\input{dcop0.tikz}\stackrel{\textup{Z1}}{=}\label{Z1}\input{dcop1.tikz}$&$\quad$&$\input{dun0.tikz}\stackrel{\textup{Z2}}{=}\label{Z2}\input{dun1.tikz}$&$\quad$&$\input{hadd0.tikz}\stackrel{\textup{W}}{=}\input{hadd1.tikz}$\\[0.5cm]
		$\input{caspider0.tikz}\stackrel{\textup{S}}{=}\label{S}\input{caspider1.tikz}$&$\quad$&$\input{cawire0.tikz}\stackrel{\textup{W1}}{=}\label{W1}\input{cawire1.tikz}\stackrel{\textup{W2}}{=}\label{W2}\input{cawire2.tikz}$&$\quad$&$\input{cahadp0.tikz}\stackrel{\textup{H'}}{=}\label{H'}\input{cahadp1.tikz}$\\[0.5cm]
		$\input{rcopy0.tikz}\stackrel{\textup{c}}{=}\input{rcopy1.tikz}$&$\quad$&$\input{bigebre0.tikz}\stackrel{\textup{b}}{=}\input{bigebre1.tikz}$&$\quad$&$\input{bone0.tikz}\stackrel{\textup{s2}}{=}\input{bone1.tikz}$\\[0.5cm]
		
		$\input{exp0.tikz}\stackrel{\textup{P}}{=}\input{exp1.tikz}$&$\quad$&$\input{euler0.tikz}\stackrel{\textup{e}}{=}\input{euler1.tikz}$&$\quad$&\begin{tabular}{l}
					$\beta_{1}=\arg(z)+\arg(z')$\\[-0.2cm]
					$\beta_{2}=2\arg(i+|\frac{z}{z'}|)$\\[-0.2cm]
					$\beta_{3}=\arg(z)-\arg(z')$\\[-0.2cm]
					$\gamma=x^{+}-\arg(z)+\frac{\pi-\beta_2}{2}$
				\end{tabular}\\
		\hline
	\end{tabular}
	\caption{\label{fig:compact}{\bf Axioms of the \szx-calculus}, where $x^{+}\coloneqq\frac{\alpha_1 + \alpha_2}{2}$, $x^{-}\coloneqq x^{+}- \alpha_2$, $z\coloneqq-\sin(x^{+})+ i\cos(x^{-})$, $z'\coloneqq\cos(x^{+})-i\sin(x^{-})$ and $z'=0 \Rightarrow \beta_2 =0$.
	In the spider fusion rule, there must be at least one wire between the spiders annotated by $\alpha$ and $\beta$. The colour-swapped versions of those rules also hold. The bold font wires stand for wires of any size $n\geq 1$.
	}
\end{figure}

\begin{lemma}\label{lm:caxiom} 
	All the rules of the \szx-calculus can be derived from the compact axioms of Figure \ref{fig:compact} together with the \textit{Only Topology Matters} paradigm.
\end{lemma}

\noindent The proof is in the appendix at page \pageref{pr:caxiom}.

\section{Axiomatising binary matrices for compressing diagrams}

In this section, we introduce a new generator for the \szx-calculus, parametrized by a binary matrix, allowing us to represent large graphical structures in a compact way: $\forall n, m\in \mathbb N^*$, $\forall A\in \mathbb F_2^{m\times n}$, $\input{mat0-def.tikz}:1_n\to 1_m$. All-ones matrices will be omitted:  $\input{mat1-def.tikz} \coloneqq  \input{mat-all1-def.tikz}$ where $\forall i,j, J_{i,j}=1$.  
The new generator is interpreted as follows: \[\forall A\in \mathbb F_2^{m\times n},\interp{\input{mat0-def.tikz}}_s = (\ket x \mapsto \ket {Ax}, 1_n,1_m)\] where  the matrix product $Ax$ is in $\mathbb F_2$ and $x$ is seen as a column vector i.e.~$(Ax)_{i} = \sum_{k=1}^nA_{i,k}x_k \bmod 2$.  

\begin{remark} Note that, compared to \cite{chancellor2016graphical}, the matrix is not necessarily connected to green and red spiders.
It is therefore a more elementary generator.
\end{remark} 

Those matrices are required to satisfy the four axioms given in Figure \ref{fig:matrix}, which are sound. 

\begin{figure}[!h]
	\begin{tabular}{|cccc|}\hline 
		$\input{zero0.tikz}\stackrel{\textup{0}}{=}\label{0}\input{zero1.tikz}$&$\qquad\input{one0.tikz}\stackrel{\textup{1}}{=}\label{1}\input{one1.tikz}$&$\qquad\input{raw0.tikz}\stackrel{\textup{L}}{=}\label{L}\input{raw1.tikz}$&$\qquad\input{col0.tikz}\stackrel{\textup{C}}{=}\label{C}\input{col1.tikz}$\\[0.2cm]\hline 
	\end{tabular}
	\caption{\label{fig:matrix}Axioms for matrices, where $A\in \mathbb{F}_2^{a\times n}$, $B\in \mathbb{F}_2^{b\times n}$, $C\in  \mathbb{F}_2^{m\times c}$ and $D\in \mathbb{F}_2^{m\times d}$. $\SmallMatrix{A\\B}$ and $\SmallMatrix{C&D}$ are block matrices. }
\end{figure}

\begin{remark}  The rules of the ZX-calculus define a scaled Hopf algebra between the green and red structure. This algebra is commutative and cocommutative with a trivial antipode.  Thus, following the work of \cite{zanasi2018interacting}, the notion of $\{0,1\}$-matrices naturally emerges. It is worth noticing that it coincides with the matrices we are introducing in this section. Notice however that our axiomatisation of the matrices strongly relies on their interaction with the divider and the gatherer, which are not present in \cite{zanasi2018interacting}. 
\end{remark}

 In the following, the \szx-calculus refers to the \szx-calculus augmented with the matrix generators and the axioms of Figure \ref{fig:matrix}.

Useful equations can be derived. First, matrices are copied and erased by green nodes.

\begin{lemma}\label{lm:vcopyerase}For any $A\in \mathbb F_2^{m\times n}$, \szx~$\vdash	 \input{copy0.tikz}\stackrel{\textup{K}}{=}\label{K}\label{mcopy}\input{copyAA.tikz}$ and \szx~$\vdash\input{erase0.tikz}\stackrel{\textup{G}}{=}\label{G}\label{merase}\begin{tikzpicture}
	\begin{pgfonlayer}{nodelayer}
		\node [style=none] (0) at (0.2, 0) {};
		\node [style=ggreen] (1) at (1, 0) {};
	\end{pgfonlayer}
	\begin{pgfonlayer}{edgelayer}
		\draw [style=very thick] (0.center) to (1);
	\end{pgfonlayer}
\end{tikzpicture}
$

\end{lemma}

\noindent The proof is in the appendix at page \pageref{pr:vcopyerase}.

We define backward matrices as follows: $ \input{mat0-back.tikz}\coloneqq\input{transmat1.tikz}$.

\begin{lemma}\label{Had}\label{lm:had} $\forall A\in \mathbb F_2^{m{\times} n}$, \szx~$\vdash 
	\input{had0.tikz}\stackrel{\textup{H}}{=}\label{H}\input{had1.tikz}$ where $A^t$ is the transpose of $A$. 
\end{lemma}

\noindent The proof is in the appendix at page \pageref{pr:had}.

As a consequence, conjugating by Hadamard (\!\!\!\begin{tikzpicture}
	\begin{pgfonlayer}{nodelayer}
		\node [style=ghad] (0) at (0.5, 0) {};
		\node [style=none] (1) at (0, 0) {};
		\node [style=none] (2) at (1, 0) {};
	\end{pgfonlayer}
	\begin{pgfonlayer}{edgelayer}
		\draw [style=very thick] (1.center) to (0);
		\draw [style=very thick] (0) to (2.center);
	\end{pgfonlayer}
\end{tikzpicture}
\!\!\!) reverses the orientation and transposes the matrix (up to scalars). Since conjugating by Hadamard colour-swaps the spiders and preserves the other generators of the language, one can derive from any equation a new one (up to scalars) which consists in colour-swapping the spiders, transposing the matrices and then changing their orientation. For instance Lemma \ref{lm:vcopyerase} gives that matrices are cocopied and coerased by red nodes:

\begin{lemma}\label{lm:rcocopycoerase}
	For any $A\in \mathbb F_2^{m\times n}$, \szx~$\vdash	 \input{copy0r.tikz}\stackrel{\textup{J}}{=}\label{J}\label{cocopy}\input{copyAA-r.tikz}$ and \szx~$\vdash\input{rerase0.tikz}\stackrel{\textup{F}}{=}\label{F}\input{rerase1.tikz}$

\end{lemma}

Basic matrix operations like addition and multiplication (in $\mathbb F_2$) can be implemented graphically:

\begin{lemma}\label{lm:addmul} For any $ A,B \in \mathbb F_2^{m\times n}$, and any $C\in  \mathbb F_2^{k\times m}$,  \szx~$\vdash \input{add0.tikz}\stackrel{\textup{p}}{=}\label{p}\input{add1.tikz}$ and  \szx~$\vdash\input{mul0.tikz}\stackrel{\textup{m}}{=}\label{m}\input{mul1.tikz}$.
\end{lemma}

\noindent The proof is in the appendix at page \pageref{pr:addmul}.

Whereas all the previous properties about matrices are angle-free, some spiders whose angles are multiple of $\pi$ can be pushed through matrices as follows:

\begin{lemma}\label{pi}\label{lm:pivect} 
For any $A \in \mathbb F_2^{m\times n}$, any $v \in \mathbb F_2^{n}$ and any $u \in \mathbb F_2^{m}$,\\ \centerline{\szx~$\vdash \input{pivect0.tikz}\stackrel{\textup{N}}{=}\label{N}\input{pivect1.tikz}$ ~~~~~and~~~~~ \szx~$\vdash \input{pivect1-d.tikz}\stackrel{\textup{O}}{=}\label{O}\input{pivect0-d.tikz}$}
\end{lemma}
\noindent The proof is in the appendix at page \pageref{pr:pivect}.

Injective matrices enjoy some specific properties: 

 \begin{lemma}\label{lm:inj}\label{inj} For any $A\in \mathbb F_2^{m\times n}$, the following properties are equivalent:\\
		\begin{tabular}{ll}
			$(1)$ $A$ is injective. &$(3)$   \szx~$\vdash\input{isur0.tikz}\stackrel{\textup{I1}}{=}\label{I1}\begin{tikzpicture}
	\begin{pgfonlayer}{nodelayer}
		\node [style=none] (77) at (0, 0) {};
		\node [style=none] (92) at (1.5, 0) {};
	\end{pgfonlayer}
	\begin{pgfonlayer}{edgelayer}
		\draw [style=very thick] (77.center) to (92.center);
	\end{pgfonlayer}
\end{tikzpicture}
$ \\ $(2)$  \szx~$\vdash\input{rerase0g.tikz}=\input{rerase1-g.tikz}$&
			  $(4)$  \szx~$\vdash\input{copyAA-g.tikz}=\input{copy0g.tikz}$
		\end{tabular}
\end{lemma}
\noindent The proof is in the appendix at page \pageref{pr:inj}.

By Hadamard conjugation, we obtain some dual properties for surjective matrices:

 \begin{lemma}\label{lm:surj}\label{surj} For any $A\in \mathbb F_2^{m\times n}$, the following properties are equivalent:\\
 	\begin{tabular}{ll}
 		$(1)$ $A$ is surjective. &$(3)$   \szx~$\vdash\input{isur0-d.tikz}\stackrel{\textup{S1}}{=}\label{S1}\begin{tikzpicture}
	\begin{pgfonlayer}{nodelayer}
		\node [style=none] (77) at (0, 0) {};
		\node [style=none] (92) at (1.5, 0) {};
	\end{pgfonlayer}
	\begin{pgfonlayer}{edgelayer}
		\draw [style=very thick] (77.center) to (92.center);
	\end{pgfonlayer}
\end{tikzpicture}
$ \\ $(2)$  \szx~$\vdash\input{erase0-d.tikz}=\begin{tikzpicture}
	\begin{pgfonlayer}{nodelayer}
		\node [style=none] (0) at (1, 0) {};
		\node [style=ggreen] (1) at (0.2, 0) {};
	\end{pgfonlayer}
	\begin{pgfonlayer}{edgelayer}
		\draw [style=very thick] (0.center) to (1);
	\end{pgfonlayer}
\end{tikzpicture}
$&
 		$(4)$  \szx~$\vdash\input{copy0r-d.tikz}=\input{copyAA-d.tikz}$
 	\end{tabular}
\end{lemma}

Due to the universality of the \szx-calculus, matrices are expressible as \szx-diagrams, and the matrix generator \!\!\input{mat0-def.tikz}\!\! is actually a compact representation of a  green/red bipartite graphs whose biadjacency matrix is $A$:

\begin{lemma}\label{lm:rep} For any $A \in \mathbb F_2^{m\times n}$, \szx~$\vdash\input{rep0.tikz}\stackrel{\textup{B}}{=}\label{B}\input{rep1.tikz}$
	where  $A$ represents in the RHS diagram the adjacency matrix of the bipartite green/red graph, and $|A|$ is the number of $1$ in $A$.
	
\end{lemma}

\noindent The proof is in the appendix at page \pageref{pr:rep}.

Lemma \ref{lm:rep} and \ref{thm:sem} imply the completeness of the \szx-calculus with matrices: 

\begin{theorem}
 \szx-calculus with matrices is complete. 
\end{theorem}

\section{Applications}\label{sec:applications}

This section provides two examples of the \szx-calculus in action.

\subsection{Application to graph states}

Graph states \cite{hein2006entanglement} form a subclass of quantum states that can be represented by simple undirected graphs where each vertex represents a qubit and the edges represent intuitively the entanglement between qubits. The graph state formalism is widely used in quantum information processing, providing combinatorial characterisations of quantum properties in measurement-based quantum computing \cite{Raussendorf_2003,DKPP09,gflow}, secret sharing \cite{MS08,KMMP09,Javelle:2013fk}, error correcting codes \cite{Schlingemann:2001fk,bell2014experimental} etc. Graph states are also strongly related to the ZX-calculus \cite{duncan2009graph} where they have been used for instance in proving the completeness of some fragments \cite{pi_2-complete,pivoting}.  
 A graph state is a particular kind of stabilizer state and thus can be defined as  a fixpoint: given a graph $G$ of order $n$, the corresponding graph state $\ket G$ is a the unique state (up to a global phase) such that for any vertex $u$, applying $X=\interp{\begin{tikzpicture}
	\begin{pgfonlayer}{nodelayer}
		\node [style=rdot] (0) at (0.6, 0) {\footnotesize$\pi$};
		\node [style=none] (1) at (0, 0) {};
		\node [style=none] (2) at (1.2, 0) {};
	\end{pgfonlayer}
	\begin{pgfonlayer}{edgelayer}
		\draw (1.center) to (0);
		\draw (0) to (2.center);
	\end{pgfonlayer}
\end{tikzpicture}
}$ on $u$ and $Z=\interp{\begin{tikzpicture}
	\begin{pgfonlayer}{nodelayer}
		\node [style=gdot] (0) at (0.6, 0) {\footnotesize$\pi$};
		\node [style=none] (1) at (0, 0) {};
		\node [style=none] (2) at (1.2, 0) {};
	\end{pgfonlayer}
	\begin{pgfonlayer}{edgelayer}
		\draw (1.center) to (0);
		\draw (0) to (2.center);
	\end{pgfonlayer}
\end{tikzpicture}
}$ on its neighbours leaves the state unchanged. The global phase is fixed by the extra condition $\langle 0^n\ket G = \frac{1}{\sqrt{n}}$.

A graph state admits a simple representation as a ZX-diagram: each vertex is represented by a green spider connected to an output, and each edge is represented by a Hadamard (\!\!\!\begin{tikzpicture}
	\begin{pgfonlayer}{nodelayer}
		\node [style=had] (0) at (0.5, 0) {};
		\node [style=none] (1) at (0, 0) {};
		\node [style=none] (2) at (1, 0) {};
	\end{pgfonlayer}
	\begin{pgfonlayer}{edgelayer}
		\draw (1.center) to (0);
		\draw (0) to (2.center);
	\end{pgfonlayer}
\end{tikzpicture}
\!\!\!) connecting the corresponding green dots. In the following, we provide two alternative, scalable, representations of graph states: the first  is a compact matrix-based  representation of bipartite graph states, the second  is an inductive definition of arbitrary graph state,  allowing inductive proofs.  In both representations, we provide diagrammatic proofs of some key properties of the graph states.

First, any bipartite graph state can be depicted with a \szx-diagram via its biadjacency matrix:

\begin{lemma}\label{lm:bgrphs}
For any bipartite graph $G$ with biadjacency matrix $\Gamma\in \mathbb F_2^{m\times n}$, 

\centerline{$\interp{{\input{bgbox1.tikz}}}_s = (|G\rangle,0,1_{n+m})$}
\end{lemma}

\begin{proof}[\textbf{Proof of Lemma \ref{lm:bgrphs}}]\phantomsection\label{pr:bgrphs}

	The last two  components are straightforward typing, for the first one we use the characterization of $|G\rangle$ by its stabilizer \cite{hein2006entanglement}.
	$|G\rangle$ is the unique (up to a scalar) common fix point of $X_{u}Z_{N_{u}}$ for all vertices $u$ of $G$. Each subset of vertices is identified with its characteristic vector, e.g. $\mathbf{N}_u\pi$ is a vector with a $\pi$ at the position $i$ if the $i$-$th$ vertex is a neighbour of $u$, and $0$ otherwise. The following proof uses the fact that $\Gamma^t {u}=\mathbf{N}_u$. We assume $u$ is in the first part of the bipartite graph (the other case is similar):

		\noindent	\input{bgsproofa0.tikz}$\stackrel{\hyperref[X]{X},\hyperref[Z]{Z}}{=}$\input{bgsproofa1.tikz}$\stackrel{\hyperref[h]{h}}{=}$\input{bgsproofa2.tikz}$\stackrel{\hyperref[O]{O}}{=}$
			\input{bgsproofa3.tikz}$\stackrel{\hyperref[s1]{s1}}{=}$\input{bgsproofa4.tikz}
	
	\vspace{0.2cm}
	\noindent It remains to take care of the scalar. We see that $\sqrt{2^{n+m}} \langle 0^{n+m}\ket G=1$:
	
\input{bgsproofb0.tikz}$\stackrel{\hyperref[X]{X},\hyperref[Z]{Z}}{=}$\input{bgsproofb1.tikz}$\stackrel{\hyperref[h]{h}}{=}$\input{bgsproofb2.tikz}$\stackrel{\hyperref[G]{G}}{=}$
			\input{bgsproofb3.tikz}$\stackrel{\hyperref[s2]{s2}}{=}$\input{bgsproofb4.tikz}
\end{proof}

A fundamental property of graph states is that graph transformations (like pivoting and local complementation) can be performed on graph states using local operations. Given a bipartite graph $G$, pivoting according to an edge $(u,v)$ produces a graph denoted $G\wedge uv$ where the labels $u$ and $v$ are exchanged and their neighbourhood is complemented: for any $w\in \mathbf{N}_u\setminus v$ and $t\in \mathbf{N}_v\setminus u$, $w$ and $t$ are connected in $G\wedge uv$ iff they are not connected in $G$.  Pivoting can be implemented on bipartite graph states by simply applying Hadamard on $u$ and $v$: $H_{u,v}\ket{G}=\ket{G\wedge uv}$ \cite{VdN05,mhalla2012graph}. This property can be derived in the \szx-calculus, and its proof, given in the appendix at page \pageref{pr:piv} due to space limits, provides an interesting example of the \szx-calculus in action:

\begin{lemma}\label{lm:piv}Given a bipartite graph $G$ and an edge $(u,v)$,

\vspace{0.1cm}
\centerline
{$\szx\vdash \input{pivot1.tikz}=\input{pivot0.tikz}$}
\vspace{0.1cm}

\noindent where $\Gamma_G$ (resp. $\Gamma_{G\wedge uv}$) is the biadjacency matrix of $G$ (resp. $G\wedge uv$) such that $u$ corresponds to the first row (resp. column) and $v$ to the first column (resp. row). 

\end{lemma}

\noindent The proof is in the appendix at page \pageref{pr:piv}.

Now we introduce a general inductive definition of graph state boxes, which associates a \szx-diagram with any (not necessarily bipartite) graph.

\begin{definition}\label{def}
	Given a graph $G$ with ordered vertices, the corresponding graph state box is defined by:\\
	\vspace{-0.5cm}
	\begin{equation*}
	\begin{tikzpicture}
	\begin{pgfonlayer}{nodelayer}
		\node [style=box] (0) at (0, 0) {\footnotesize$K_1$};
		\node [style=none] (1) at (1.5, 0) {};
	\end{pgfonlayer}
	\begin{pgfonlayer}{edgelayer}
		\draw [in=180, out=0] (0) to (1.center);
	\end{pgfonlayer}
\end{tikzpicture}
\coloneqq\input{gdefb1.tikz}~~and ~~~\begin{tikzpicture}
	\begin{pgfonlayer}{nodelayer}
		\node [style=box] (0) at (0, 0) {\footnotesize$G$};
		\node [style=none] (1) at (1, 0) {};
	\end{pgfonlayer}
	\begin{pgfonlayer}{edgelayer}
		\draw [style=very thick] (1.center) to (0);
	\end{pgfonlayer}
\end{tikzpicture}
\coloneqq\input{ggsdef1.tikz}
	\end{equation*}
	where $K_1$ is the graph of order $1$, $u$ is the first vertex of $G$, $\tau$ is a permutation on the list of vertices of $G\backslash u$ which puts the neighbourhood of $u$ first and then the other vertices. 
\end{definition}

\begin{lemma}\label{lm:grphs}
	$\left\llbracket\begin{tikzpicture}
	\begin{pgfonlayer}{nodelayer}
		\node [style=box] (4) at (0, 0) {\footnotesize$G$};
		\node [style=none] (5) at (1, 0) {};
	\end{pgfonlayer}
	\begin{pgfonlayer}{edgelayer}
		\draw [style=very thick] (5.center) to (4);
	\end{pgfonlayer}
\end{tikzpicture}
\right\rrbracket_{s}=(|G\rangle,0,1_{n})$.
\end{lemma}

\noindent The proof is in the appendix at page \pageref{pr:grphs}. 

We will now use the \szx-calculus to show the property known as local complementation. Given a vertex $u$ of a graph $G$ the local complementation of $G$ according to $u$ is the graph $G\star u$ which is $G$ where all edges between neighbours of $u$ have been complemented, that is, edges became non-edges and non-edges became edges.

\begin{theorem}\label{thm:locomp} For any graph $G$ and vertex $u$,  
	$\szx \vdash \begin{tikzpicture}
	\begin{pgfonlayer}{nodelayer}
		\node [style=box] (1) at (0, 0) {\footnotesize$G\star u$};
		\node [style=none] (2) at (2, 0) {};
	\end{pgfonlayer}
	\begin{pgfonlayer}{edgelayer}
		\draw [style=very thick] (1) to (2.center);
	\end{pgfonlayer}
\end{tikzpicture}
=\input{locomp1.tikz}$.
\end{theorem}

\noindent The proof is in the appendix at page \pageref{pr:locomp}.

\subsection{Application to error correcting codes}

The original motivation for the development of a scalable ZX-calculus was the design of tripartite Coherent Parity Checking (CPC) error correcting codes \cite{chancellor2016graphical}. We reformulate here in the \szx-calculus  the definition of those codes and the proof of some elementary properties.

The idea is to spread the information of some logical qubits over a bigger number of physical qubits. In our example the code is parametrized by three matrices $B\in \mathbb{F}_2^{a \times b}$, $P\in \mathbb{F}_2^{c \times b}$ and $C\in \mathbb{F}_2^{c \times a}$. The aim is to encode $b$ logical qubits into $a{+}b{+}c$ physical ones.

\begin{definition}
	The tripartite CPC encoder $E:1_b\to 1_a{+}1_b{+}1_c$ and decoder $D :1_a{+}1_b{+}1_c \to 1_b$ defined by the matrices $B\in \mathbb{F}_2^{a \times b}$, $P\in \mathbb{F}_2^{c \times b}$ and $C\in \mathbb{F}_2^{c \times a}$ are:
	\begin{center}
		\begin{tabular}{ccc}
			$E\coloneqq\input{cpcdf.tikz}$&$\quad\quad$&$D\coloneqq\input{ddcpcdf.tikz}$\\[-0.2cm]
		\end{tabular}\\
	\end{center}
\end{definition}

We can prove that the code is correct when there are no errors, in other words:

\begin{lemma}\label{lm:iso}
	The encoder is an isometry that is $\szx\vdash D\circ E=\begin{tikzpicture}[scale=0.8]
	\begin{pgfonlayer}{nodelayer}
		\node [style=none] (9) at (0, 0) {};
		\node [style=none] (11) at (1, 0) {};
	\end{pgfonlayer}
	\begin{pgfonlayer}{edgelayer}
		\draw [style=very thick] (9.center) to (11.center);
	\end{pgfonlayer}
\end{tikzpicture}
$.
\end{lemma}
\noindent The proof is in appendix at page \pageref{pr:iso}.

We now end by showing what happens when errors go through the decoder: $x$, $y$ and $z$ (resp. $x'$, $y'$, $z'$) are vectors of phase flip errors (resp. bit flip errors). The implementation of the decoder involves some  measurements, which according to Lemma \ref{lm:err}, produce some syndromes ($|x| =\sum_ix_i \bmod 2$, $z{+}Cx {+} Py$, $x'{\!+}y'{\!+}C^t z' {\!+} BP^t z'$, and $|z'|$) which guide us to correct the middle wire. Of course the exact protocol and its efficiency depend on clever choices of $B$, $P$ and $C$, see \cite{chancellor2016graphical} for details.

\begin{lemma}\label{lm:err}
	The following equalities hold in the \szx-calculus:
	\vspace{0.2cm}
	
		\centerline{	$\scalebox{0.91}{\input{verdec0.tikz}}\!\!=\!\!\scalebox{0.91}{\input{verdec1.tikz}}\qquad~~$
			$\scalebox{0.91}{\input{rerdec0.tikz}}\!\!=\!\!\scalebox{0.91}{\input{rerdec1.tikz}}$}

\end{lemma}

\vspace{-0.1cm}
\noindent The proof is in appendix at page \pageref{pr:err}.

\vspace{-0.1cm}

\section{Conclusion  and further work}
\vspace{-0.1cm}

We have introduced the \szx-calculus, a formal and compact graphical language for quantum reasoning, and proved its universality, soundness, and completeness. This work is addressing two main objectives.
First, to demonstrate that some of the ingredients for scalability which were sketched out in \cite{chancellor2016graphical} -- like the thick wires and the use of matrices -- together with some new ingredients  -- like the divider and the gatherer -- can be axiomatised to provide a complete scalable graphical language.  
Our second objective was to provide a sufficiently precise definition of the language to consider an implementation in a graphical proof assistant like Quantomatic. This last point would pave the way towards the formal verification of large scale quantum protocols and algorithms.

We aim to provide a language ready for applications and available to most of the quantum computing community. For this reason, we have deliberately avoided a categorical presentation. % which would have been counter productive. 
A fully categorical description of the scalable construction will be the subject of further work. We nevertheless provide here a sketch of how our construction can be generalised in a categorical setting. Graphical languages can be defined as props, see \cite{baez2017props} and \cite{zanasi2018interacting}, that is symmetric strict monoidal categories whose set of objects is freely generated by one object we denote $1$. In fact it is possible to define a scalable construction for any coloured prop. Given a set $C$  of colours we can define two $\langle\langle C \rangle-\varepsilon\rangle$-coloured props $\mathbb{D}_C$ and $\mathbb{G}_C$ whose objects are formal sums of $1_{n_c}$ and morphisms are respectively generated by dividers and gatherers for each pair $(n,c)$. The elimination rule is a distribution rule as in \cite{lack2004composing}, which allows us to define the composed prop $\mathbb{D}_C ;\mathbb{G}_C$. The prop of wires $\mathbb{W}_C$ is then defined as this composition quotiented by the expansion rule. This last pro satisfies a rewire theorem similar to Theorem \ref{thm:rewire}. Then given a $C$-coloured prop $\mathbf{P}$, we define the $\langle\langle C \rangle-\varepsilon\rangle$-coloured prop $\overline{\mathbf{P}}$ which has the same generators and equations as those of $\mathbf{P}$ on wires of size $1$. Finally the scalable prop $\mathcal{S}\mathbf{P}$ is defined as the composition of prop $\mathbb{D}_C ;\overline{\mathbf{P}} ;\mathbb{G}_C$ quotiented by the expansion rule. The corresponding distribution rules follows the same pattern as in \ref{dist}. Such a generalization gives scalable versions of any graphical language based on props such as the $ZW$--calculus \cite{hadzihasanovic2017algebra}, the $ZH$-calculus \cite{backens2018zh} or $\mathbb{IH}$ \cite{bonchi2017interacting}.

%\bibliography{scal}

\appendix

\section{Appendix}

\begin{proof}[\textbf{Proof of Theorem \ref{thm:rewire}}]\phantomsection\label{pr:rewire}

	The proof goes by first defining expanded forms and showing that any diagram can be put into such a form. Then we show that any diagram of type $\sum\limits_i 1_{n_i} \to \sum\limits_j 1_{m_j}$ is equal to the same expanded diagram.
	
	The situation of a wire is defined as a pair of elements in the set $\{i,o,d,g\}$, where $i$ stands for input, $o$ for output, $d$ for divider, and $g$ for gatherer. For example, a wire which links an input to an output has situation $(i,o)$ and a wire linking a gatherer to a divider has situation $(g,d)$. We only consider non-identity dividers and gatherers. The possible situations for a wire of size $1$ are : $(i,o)$, $(i,g)$, $(d,o)$, and $(d,g)$. The possible situations for a big wire are the same additioned to $(i,d)$, $(d,d)$, $(g,o)$, $(g,g)$ and $(g,d)$.
	
	We say that a diagram is expanded when all the big wires are in situation $(i,d)$, $(d,d)$, $(g,o)$ or $(g,g)$.
	
	Any diagram can be rewritten into an expanded one:
	
	When a big wire is in situation $(i,o)$, $(i,g)$, $(d,o)$ or $(d,g)$, then the expansion rule can be applied decreasing strictly the size of the wires in a bad situation. Applying it recursively we end up with a diagram where no big wires are in situation $(i,o)$, $(i,g)$, $(d,o)$ or $(d,g)$.
	
	If a big wire is in situation $(g,d)$, then the elimination rule applies and strictly decreases the size of the wires. So we can apply it until there are no big wires in situation $(g,d)$.
	
	Finally we obtain a diagram in expanded form.
	
	We define inductively $d_n:1_{n}\to n_{1}$ by $d_{0}=I$, and $d_n=\left(id_1 \otimes d_{n-1}\right)\circ divide_{n}$. $d_n$ is a sequence of $n-1$ dividers of decreasing size.
	By symmetry we also define inductively $g_{m}:m_{1}\to 1_{m}$ by $d_{0}=I$, and $g_m=gather_{m}\circ \left(id_1 \otimes g_{m-1}\right)$, which is a sequence of $m-1$ gatherers of increasing size.
	
	Given a diagram $\omega:\sum\limits_i 1_{n_i} \to \sum\limits_j 1_{n_j}$, we put it into expanded form. We see that by sliding the gatherers and dividers along the wires we obtain an expanded diagram of the form $\Gamma\circ\Delta$ where $\Delta$ contains no gatherers and $\Gamma$ contains no dividers. Let's suppose that there is a big wire linking $\Delta$ and $\Gamma$. Then it would either be in situation $(i,o)$ or $(d,g)$ but such situations are not possible in an expanded diagram. We conclude that $\Delta$ is of type $\sum\limits_i 1_{n_i} \to \left(\sum\limits_i n_i\right)$ and $\Gamma$ of type $\left(\sum\limits_i n_i\right)\to \sum\limits_j 1_{n_j}$. In $\Delta$, the only allowed situations for a small wire are $(i,o)$ and $(d,o)$, and for a big wire they are $(i,d)$ and $(d,d)$. This enforces a unique structure: $\Delta=\bigotimes\limits_{i}d_{n_i}$, and by symmetry $\Gamma=\bigotimes\limits_{j}g_{m_j}$. Finally $\omega=\left(\bigotimes\limits_{j}g_{m_j}\right)\circ \left(\bigotimes\limits_{i}d_{n_i}\right)$.
	
\end{proof}

\begin{proof}[\textbf{Proof of Theorem \ref{thm:univ}}]\phantomsection\label{pr:univ}
	Let $\left(M,a,b\right)$ be a valid triple. Since the ZX-calculus is universal for matrices we can take a ZX diagram with interpretation $M$, and embed it in the \szx-calculus with wire of size $1$. This gives a diagram with interpretation $\left(M,S(a),S(b)\right)$. Using the notation of \ref{pr:rewire}, we then precompose by $\bigotimes\limits_{i}d_{m_i}$ and post compose by $\bigotimes\limits_{j}g_{n_j}$ where $\bigotimes\limits_{i}d_{m_i}\coloneqq a$ and $\bigotimes\limits_{j}g_{n_j}\coloneqq b$. This provides a diagram with interpretation $\left(M,a,b\right)$.
	
\end{proof}

\begin{proof}[\textbf{Proof of Lemma \ref{lm:srule}}]\phantomsection\label{pr:srule}
	Starting from a rule with generators of size $n$, the idea is to use the dividers to obtain $n$ copies of the rule, then applying it $n$ times we can come back to generators of size $n$. We only do the proof for the copy rule; the other rules follow the same pattern.
	
	We want to prove: $\input{nrcopy0.tikz}=\begin{tikzpicture}
	\begin{pgfonlayer}{nodelayer}
		\node [style=gred] (2) at (0, 0.5) {};
		\node [style=none] (3) at (1, 0.5) {};
		\node [style=none] (5) at (1, -0.5) {};
		\node [style=gred] (9) at (0, -0.5) {};
	\end{pgfonlayer}
	\begin{pgfonlayer}{edgelayer}
		\draw [style=very thick] (2) to (3.center);
		\draw [style=very thick] (5.center) to (9);
	\end{pgfonlayer}
\end{tikzpicture}
$.
	
	By induction on $n$. If $n=1$ this the usual copy rule of the ZX-calculus. If $n>1$:
	
	\begin{center}
		\setlength\tabcolsep{1pt}
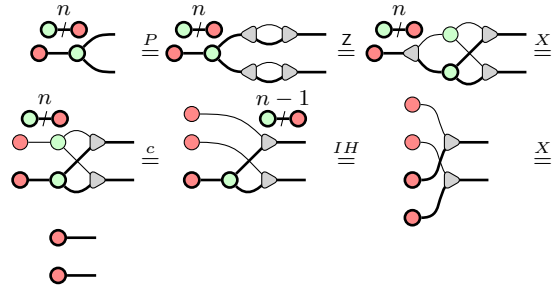
		\begin{tabular}{cccccc}

			\input{ncopyproof0.tikz}&$\stackrel{\hyperref[P]{P}}{=}$&\input{ncopyproof1.tikz}&$\stackrel{\hyperref[Z]{\textsf{Z}}}{=}$&\input{ncopyproof2.tikz}&$\stackrel{\hyperref[X]{X}}{=}$\\
			\input{ncopyproof3.tikz}&$\stackrel{\hyperref[c]{c}}{=}$&\input{ncopyproof4.tikz}&$\stackrel{IH}{=}$&\input{ncopyproof5.tikz}&$\stackrel{\hyperref[X]{X}}{=}$\\

			\input{ncopyproof6.tikz}&&&&&
		\end{tabular}
	\end{center}
	
\end{proof}

\begin{proof}[\textbf{Proof of Theorem \ref{thm:sem}}]\phantomsection\label{pr:sem}
	
	We extend the definition of expanded form \ref{pr:sem} to the \szx-calculus. We add new situations by considering $s$ the generators of size $1$ and $S$ the generators of size strictly bigger than $1$. The new possible situations for a small wire are $(i,s)$, $(d,s)$, $(s,s)$, $(s,g)$ and $(s,o)$. For a big wire the new situations are $(i,S)$, $(d,S)$, $(g,S)$, $(S,S)$, $(S,g)$, $(S,d)$ and $(S,o)$.
	
	A diagram of \szx$~$is in expanded form iff the only possible situations for big wires are $(i,d)$, $(d,d)$, $(g,o)$ or $(g,g)$.
	
	For each of the new situations for big wires, applying expansion and then a distribution rule decreases strictly the size of the wire involved in problematic situations. By induction we end up with no big wires in those situations.
	
	Considering the possible situations, any diagram in expanded form can be written $\Gamma_b \circ D\circ \Delta_a$ where $\Gamma_b$ and $\Delta_a$ are diagrams in $\mathbb{W}$ only depending on $a$ and $b$. $D$ is a diagram of the usual (size $1$) ZX-calculus. So given two diagrams of the \szx-calculus with the same type $a\to b$ they can be expanded. Then we see that they are equal iff their ZX parts are equal, which by completeness of the ZX calculus is true iff their interpretations are the same.
	
\end{proof}

\begin{proof}[\textbf{Proof of Lemma \ref{lm:caxiom}}]\phantomsection\label{pr:caxiom}
	The rules of ZX-calculus follow directly by specializing the big rules with $n=1$. It only remains to show the missing distribution rules. We start by showing skew version of the distribution rules for cups and caps which are equivalent when the expansion axiom is available. 
	
	The cup is given by:
	\begin{center}
		$\input{dcupp0.tikz}\stackrel{\hyperref[s1]{s1},\hyperref[w1]{w1}}{=}\input{dcupp1.tikz}\stackrel{\hyperref[Z]{Z}}{=}\input{dcupp2.tikz}\stackrel{\hyperref[Z]{Z}}{=}\input{dcupp3.tikz}\stackrel{\hyperref[s1]{s1},\hyperref[w1]{w1}}{=}\input{dcupp4.tikz}$.
	\end{center}
	
	From this and expansion we can derive the elimination rule:
	\begin{center}
		$\input{delim0.tikz}\stackrel{\hyperref[top]{\textup{top}}}{=}\input{delim1.tikz}=\input{delim2.tikz}\stackrel{\hyperref[P]{P}}{=}\input{delim3.tikz}=\input{delim4.tikz}\stackrel{\hyperref[top]{\textup{top}}}{=}
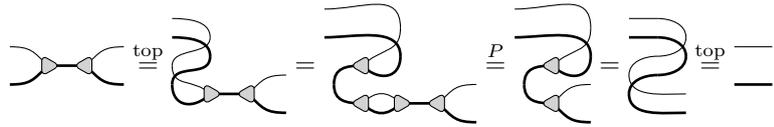
\begin{tikzpicture}
	\begin{pgfonlayer}{nodelayer}
		\node [style=none] (27) at (0, -0.5) {};
		\node [style=none] (28) at (0, 0.5) {};
		\node [style=none] (29) at (1, -0.5) {};
		\node [style=none] (30) at (1, 0.5) {};
	\end{pgfonlayer}
	\begin{pgfonlayer}{edgelayer}
		\draw [style=very thick] (29.center) to (27.center);
		\draw (28.center) to (30.center);
	\end{pgfonlayer}
\end{tikzpicture}
$.
	\end{center}
	
	Then we derive the distribution law for the cap:
	\begin{center}
		$\input{dcapp0.tikz}\stackrel{\hyperref[top]{\textup{top}}}{=}\input{dcapp1.tikz}=\input{dcapp2.tikz}\stackrel{\hyperref[P]{P}}{=}\input{dcapp3.tikz}\stackrel{\hyperref[E]{E}}{=}\input{dcapp4.tikz}$.
	\end{center}

	We have already recovered all the behaviour of dividers, gatherers and wires. The general distribution rule for green spiders follows from the unary and ternary one using the green spider rule. And finally the rules concerning red spiders follow by applying Hadamard gates on the rules concerning green spiders. 
	
\end{proof}

\begin{proof}[\textbf{Proof of Lemma \ref{lm:vcopyerase}}]\phantomsection\label{pr:vcopyerase}
	
	We start with the copy. By induction on the size $m\times n$ of $A$. If $n=m=1$ this directly follows from the rules of ZX-calculus. If $n>1$:
	
	\begin{center}
		\setlength\tabcolsep{1pt}
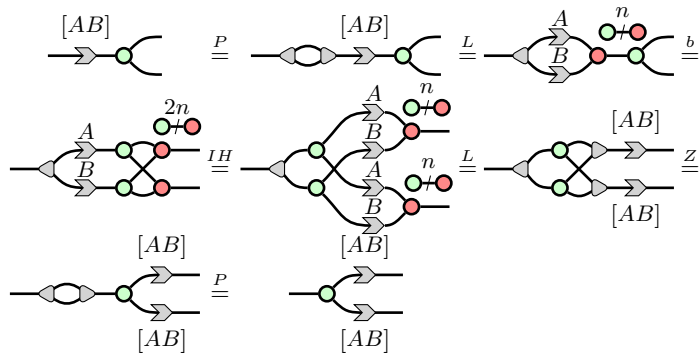
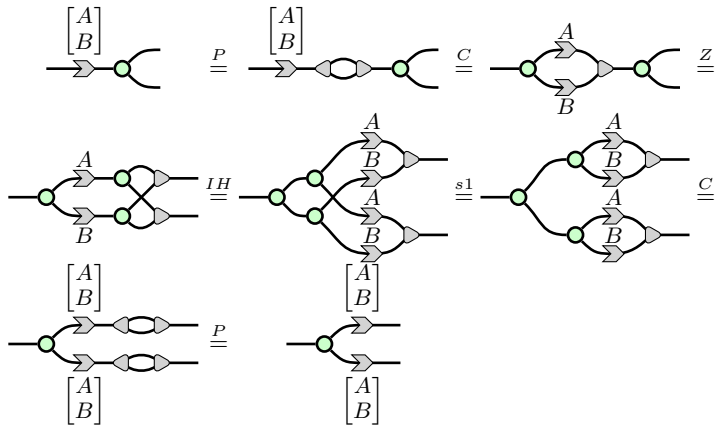
		\begin{tabular}{cccccc}
			\input{copyproofm0.tikz}&$\stackrel{\hyperref[P]{P}}{=}$&\input{copyproofm1.tikz}&$\stackrel{\hyperref[L]{L}}{=}$&\input{copyproofm2.tikz}&$\stackrel{\hyperref[b]{b}}{=}$\\
			\input{copyproofm3.tikz}&$\stackrel{IH}{=}$&\input{copyproofm4.tikz}&$\stackrel{\hyperref[L]{L}}{=}$&\input{copyproofm5.tikz}&$\stackrel{\hyperref[Z]{Z}}{=}$\\
			\input{copyproofm6.tikz}&$\stackrel{\hyperref[P]{P}}{=}$&\input{copyproofm7.tikz}&&&
		\end{tabular}
	\end{center}
	
	If $m>1$:
	
	\begin{center}
		\setlength\tabcolsep{1pt}
		\begin{tabular}{cccccc}
			\input{copyproofn0.tikz}&$\stackrel{\hyperref[P]{P}}{=}$&\input{copyproofn1.tikz}&$\stackrel{\hyperref[C]{C}}{=}$&\input{copyproofn2.tikz}&$\stackrel{\hyperref[Z]{Z}}{=}$\\
			\input{copyproofn3.tikz}&$\stackrel{IH}{=}$&\input{copyproofn4.tikz}&$\stackrel{\hyperref[s1]{s1}}{=}$&\input{copyproofn5.tikz}&$\stackrel{\hyperref[C]{C}}{=}$\\
			\input{copyproofn6.tikz}&$\stackrel{\hyperref[P]{P}}{=}$&\input{copyproofn7.tikz}&&&
		\end{tabular}
	\end{center}
	
	Then we prove erasing.
	By induction on the size $m\times n$ of $A$. If $n=m=1$ this directly follows from the rules of ZX-calculus. If $n>1$:
	
	\begin{center}
		\setlength\tabcolsep{1pt}
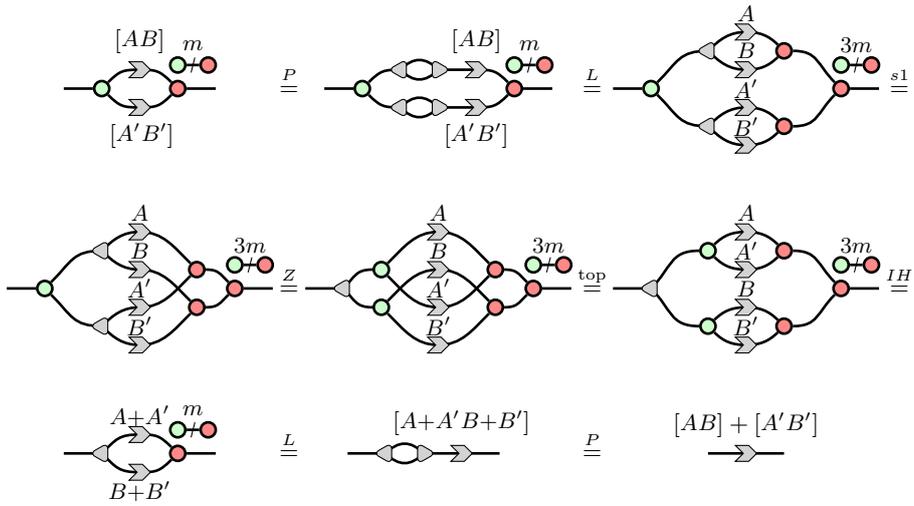
		\begin{tabular}{cccccc}
			\input{eraseproofm0.tikz}&$\stackrel{\hyperref[P]{P}}{=}$&\input{eraseproofm1.tikz}&$\stackrel{\hyperref[L]{L}}{=}$&\input{eraseproofm2.tikz}&$\stackrel{\hyperref[c]{c}}{=}$\\
			\input{eraseproofm3.tikz}&$\stackrel{IH}{=}$&\input{eraseproofm4.tikz}&$\stackrel{\hyperref[Z]{Z}}{=}$&\begin{tikzpicture}
	\begin{pgfonlayer}{nodelayer}
		\node [style=none] (111) at (0, 0) {};
		\node [style=ggreen] (113) at (1, 0) {};
	\end{pgfonlayer}
	\begin{pgfonlayer}{edgelayer}
		\draw [style=very thick] (113) to (111.center);
	\end{pgfonlayer}
\end{tikzpicture}
&
		\end{tabular}
	\end{center}
	
	If $m>1$:
	
	\begin{center}
		\setlength\tabcolsep{1pt}
		\begin{tabular}{cccccc}
			\input{eraseproofn0.tikz}&$\stackrel{\hyperref[P]{P}}{=}$&\input{eraseproofn1.tikz}&$\stackrel{\hyperref[C]{C}}{=}$&\input{eraseproofn2.tikz}&$\stackrel{\hyperref[Z]{Z}}{=}$\\
			\input{eraseproofn3.tikz}&$\stackrel{IH}{=}$&\input{eraseproofn4.tikz}&$\stackrel{\hyperref[s1]{s1}}{=}$&
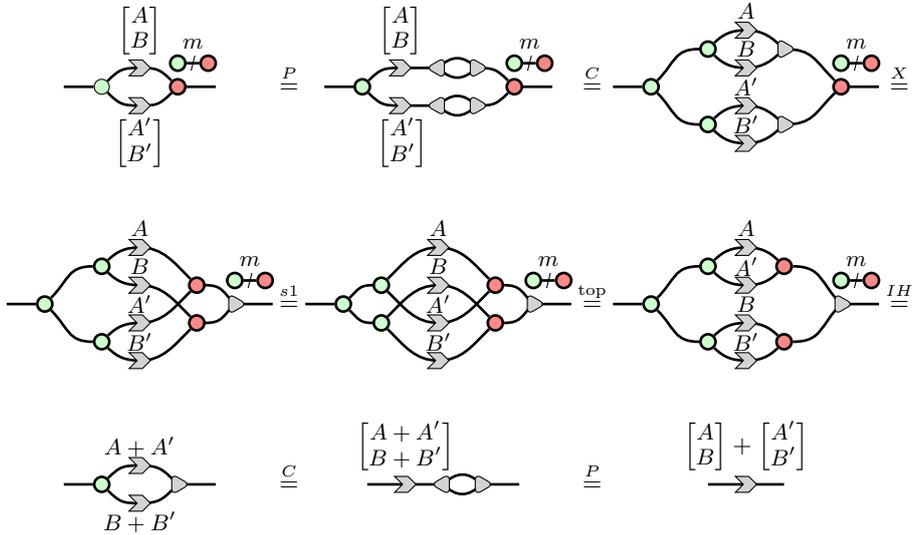
\begin{tikzpicture}
	\begin{pgfonlayer}{nodelayer}
		\node [style=none] (101) at (0, 0) {};
		\node [style=ggreen] (103) at (1, 0) {};
	\end{pgfonlayer}
	\begin{pgfonlayer}{edgelayer}
		\draw [style=very thick] (101.center) to (103);
	\end{pgfonlayer}
\end{tikzpicture}
&
		\end{tabular}
	\end{center}
	
\end{proof}

\begin{proof}[\textbf{Proof of Lemma \ref{lm:had}}]\phantomsection\label{pr:had}
	By induction on the size $m\times n$ of $A$. If $n=m=1$ this directly follows from the rules of ZX-calculus. If $n>1$:
	
	\begin{center}
		\setlength\tabcolsep{1pt}
		\begin{tabular}{cccccc}
			\input{hadproofm0.tikz}&$\stackrel{\hyperref[P]{P}}{=}$&\input{hadproofm1.tikz}&$\stackrel{\hyperref[L]{L}}{=}$&\input{hadproofm2.tikz}&$\stackrel{\hyperref[h]{h}}{=}$\\
			&&&&&\\
			\input{hadproofm3.tikz}&$\stackrel{IH}{=}$&\input{hadproofm4.tikz}&$\stackrel{\hyperref[W]{W}}{=}$&\input{hadproofm5.tikz}&$\stackrel{\hyperref[C]{C}}{=}$\\
						&&&&&\\
			\input{hadproofm6.tikz}&$\stackrel{\hyperref[P]{P}}{=}$&\input{hadproofm7.tikz}&&&
		\end{tabular}
	\end{center}
	
	If $m>1$:
	
	\begin{center}
		\setlength\tabcolsep{1pt}
		\begin{tabular}{cccccc}
			\input{hadproofn0.tikz}&$\stackrel{\hyperref[P]{P}}{=}$&\input{hadproofn1.tikz}&$\stackrel{\hyperref[C]{C}}{=}$&\input{hadproofn2.tikz}&$\stackrel{\hyperref[W]{W}}{=}$\\&&&&&\\
			\input{hadproofn3.tikz}&$\stackrel{IH}{=}$&\input{hadproofn4.tikz}&$\stackrel{\hyperref[h]{h}}{=}$&\input{hadproofn5.tikz}&$\stackrel{\hyperref[L]{L}}{=}$\\&&&&&\\
			\input{hadproofn6.tikz}&$\stackrel{\hyperref[P]{P}}{=}$&\input{hadproofn7.tikz}&&&
		\end{tabular}
	\end{center}
	
\end{proof}

\begin{proof}[\textbf{Proof of Lemma \ref{lm:addmul}}]\phantomsection\label{pr:addmul}
	We start with addition. By induction on the size $m\times n$ of $A$ and $B$.
	If $n=m=1$, this is the Hopf rule of ZX-calculus.
	If $n>1$:
	\begin{center}
		\setlength\tabcolsep{1pt}
		\begin{tabular}{cccccc}
			\input{addproofm0.tikz}&$\stackrel{\hyperref[P]{P}}{=}$&\input{addproofm1.tikz}&$\stackrel{\hyperref[L]{L}}{=}$&\input{addproofm2.tikz}&$\stackrel{\hyperref[s1]{s1}}{=}$\\&&&&&\\
			\input{addproofm3.tikz}&$\stackrel{\hyperref[Z]{Z}}{=}$&\input{addproofm4.tikz}&$\stackrel{\hyperref[top]{\textup{top}}}{=}$&\input{addproofm5.tikz}&$\stackrel{IH}{=}$\\&&&&&\\
			\input{addproofm6.tikz}&$\stackrel{\hyperref[L]{L}}{=}$&\input{addproofm7.tikz}&$\stackrel{\hyperref[P]{P}}{=}$&\input{addproofm8.tikz}&
		\end{tabular}
	\end{center}
	
	If $m>1$:
	
	\begin{center}
		\setlength\tabcolsep{1pt}
		\begin{tabular}{cccccc}
			\input{addproofn0.tikz}&$\stackrel{\hyperref[P]{P}}{=}$&\input{addproofn1.tikz}&$\stackrel{\hyperref[C]{C}}{=}$&\input{addproofn2.tikz}&$\stackrel{\hyperref[X]{X}}{=}$\\&&&&&\\
			\input{addproofn3.tikz}&$\stackrel{\hyperref[s1]{s1}}{=}$&\input{addproofn4.tikz}&$\stackrel{\hyperref[top]{\textup{top}}}{=}$&\input{addproofn5.tikz}&$\stackrel{IH}{=}$\\&&&&&\\
			\input{addproofn6.tikz}&$\stackrel{\hyperref[C]{C}}{=}$&\input{addproofn7.tikz}&$\stackrel{\hyperref[P]{P}}{=}$&\input{addproofn8.tikz}&
		\end{tabular}
	\end{center}
	
	Now we prove multiplication. By induction on the size $b\times a$ of $A$ and $c\times b$ of $B$.
	
	If $a=b=c=1$ this follows from the zero and one axioms and the scalar rule.
	If $a>1$:
	
	\begin{center}
		\setlength\tabcolsep{1pt}
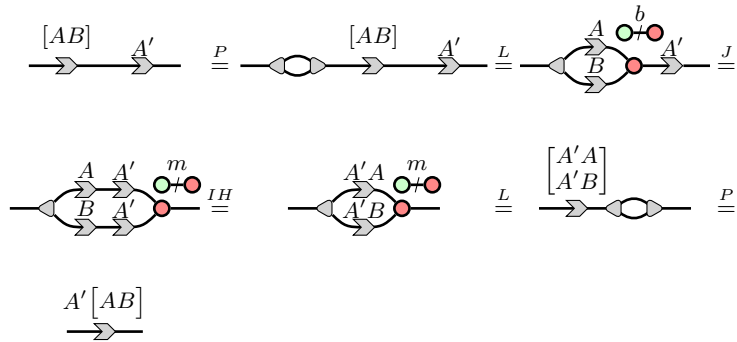
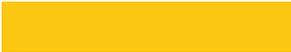
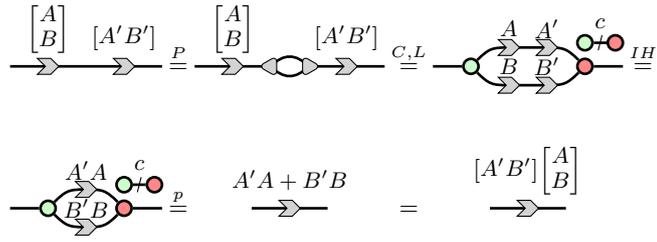
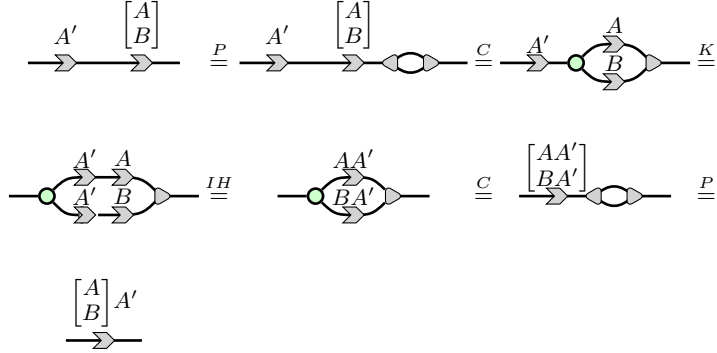
		\begin{tabular}{cccccc}
			\input{mulproofa0.tikz}&$\stackrel{\hyperref[P]{P}}{=}$&\input{mulproofa1.tikz}&$\stackrel{\hyperref[L]{L}}{=}$&\input{mulproofa2.tikz}&$\stackrel{\hyperref[J]{J}}{=}$\\&&&&&\\
			\input{mulproofa3.tikz}&$\stackrel{IH}{=}$&\input{mulproofa4.tikz}&$\stackrel{\hyperref[L]{L}}{=}$&\input{mulproofa5.tikz}&$\stackrel{\hyperref[P]{P}}{=}$\\&&&&&\\
			\input{mulproofa6.tikz}&&&&&
		\end{tabular}
	\end{center}
	
	If $b>1$:
	\begin{center}
		\setlength\tabcolsep{1pt}
		\begin{tabular}{cccccc}
			\input{mulproofb0.tikz}&$\stackrel{\hyperref[P]{P}}{=}$&\input{mulproofb1.tikz}&$\stackrel{\hyperref[C]{C},\hyperref[L]{L}}{=}$&\input{mulproofb2.tikz}&$\stackrel{IH}{=}$\\&&&&&\\
			\input{mulproofb3.tikz}&$\stackrel{\hyperref[p]{p}}{=}$&\input{mulproofb4.tikz}&$=$&\input{mulproofb5.tikz}&
		\end{tabular}
	\end{center}

	If $c>1$:
	\begin{center}
		\setlength\tabcolsep{1pt}
		\begin{tabular}{cccccc}
			\input{mulproofc0.tikz}&$\stackrel{\hyperref[P]{P}}{=}$&\input{mulproofc1.tikz}&$\stackrel{\hyperref[C]{C}}{=}$&\input{mulproofc2.tikz}&$\stackrel{\hyperref[K]{K}}{=}$\\&&&&&\\
			\input{mulproofc3.tikz}&$\stackrel{IH}{=}$&\input{mulproofc4.tikz}&$\stackrel{\hyperref[C]{C}}{=}$&\input{mulproofc5.tikz}&$\stackrel{\hyperref[P]{P}}{=}$\\&&&&&\\
			\input{mulproofc6.tikz}&&&&&
		\end{tabular}
	\end{center}
\end{proof}

\begin{proof}[\textbf{Proof of Lemma \ref{lm:pivect}}]\phantomsection\label{pr:pivect}
	By induction on the size $m\times n$ of $A$. If $n=m=1$ they are two possibilities: if $A=1$ then this is trivial, and if $A=0$ then the $\pi$ is erased by the green node and the equation also true.
	If $m>1$:
	
	\begin{center}
		\setlength\tabcolsep{1pt}
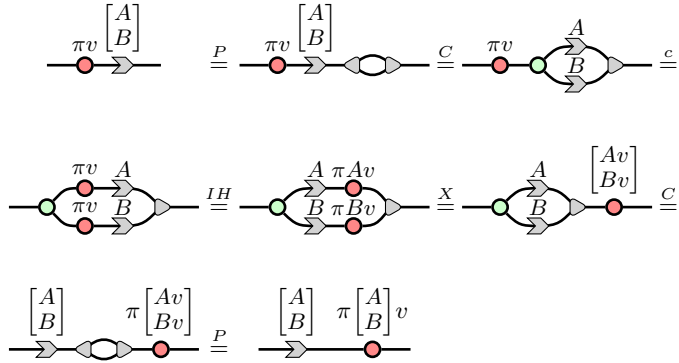
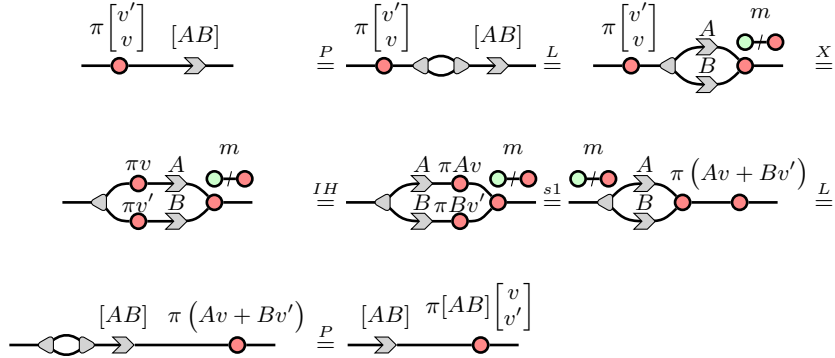
		\begin{tabular}{cccccc}
			\input{pivectproofn0.tikz}&$\stackrel{\hyperref[P]{P}}{=}$&\input{pivectproofn1.tikz}&$\stackrel{\hyperref[C]{C}}{=}$&\input{pivectproofn2.tikz}&$\stackrel{\hyperref[c]{c}}{=}$\\&&&&&\\
			\input{pivectproofn3.tikz}&$\stackrel{IH}{=}$&\input{pivectproofn4.tikz}&$\stackrel{\hyperref[X]{X}}{=}$&\input{pivectproofn5.tikz}&$\stackrel{\hyperref[C]{C}}{=}$\\&&&&&\\
			\input{pivectproofn6.tikz}&$\stackrel{\hyperref[P]{P}}{=}$&\input{pivectproofn7.tikz}&&&
		\end{tabular}
	\end{center}
	
	if $n>1$:
	
	\begin{center}
		\setlength\tabcolsep{1pt}
		\begin{tabular}{cccccc}
			\input{pivectproofm0.tikz}&$\stackrel{\hyperref[P]{P}}{=}$&\input{pivectproofm1.tikz}&$\stackrel{\hyperref[L]{L}}{=}$&\input{pivectproofm2.tikz}&$\stackrel{\hyperref[X]{X}}{=}$\\&&&&&\\
			\input{pivectproofm3.tikz}&$\stackrel{IH}{=}$&\input{pivectproofm4.tikz}&$\stackrel{\hyperref[s1]{s1}}{=}$&\input{pivectproofm5.tikz}&$\stackrel{\hyperref[L]{L}}{=}$\\&&&&&\\
			\input{pivectproofm6.tikz}&$\stackrel{\hyperref[P]{P}}{=}$&\input{pivectproofm7.tikz}&&&
		\end{tabular}
	\end{center}
	
\end{proof}

\begin{proof}[\textbf{Proof of Lemma \ref{lm:inj}}]\phantomsection\label{pr:inj}
	We show this by circular implications:\\
	
	$(1)\Rightarrow (2)$: We use the semantics. $\interp{\input{23proof0.tikz}}_s= \left(\ket{x}\mapsto \bra{Ax}\ket{0}^{\otimes m},1_n,0\right)$. The basis being orthonormal: $\bra{Ax}\ket{0}^{\otimes m}=\delta_{Ax,0}$ but since $A$ is injective $\delta_{Ax,0}=\delta_{x,0}$. Besides $\interp{\input{23proof2.tikz}}_s= \left(\ket{x}\mapsto \bra{x}\ket{0}^{\otimes n},1_n,0\right)$. So $\interp{\input{23proof0.tikz}}_s=\interp{\input{23proof2.tikz}}_s$ and by completness: $\input{23proof0.tikz}=\input{23proof2.tikz}$.\\
	
	$(2)\Rightarrow (3)$: \\
	$\input{34proof0.tikz}\stackrel{\hyperref[s1]{s1}}{=}\input{34proof1.tikz}\stackrel{\hyperref[J]{J}}{=}\input{34proof2.tikz}\stackrel{\hyperref[F]{F}}{=}\input{34proof3.tikz}\stackrel{\hyperref[s1]{s1}}{=}\input{34proof4.tikz}$.\\
	
	$(4)\Rightarrow (5)$:\\
	$\input{45proof0.tikz}\stackrel{\hyperref[K]{K}}{=}\input{45proof1.tikz}=\input{45proof2.tikz}$.\\
	
	$(5)\Rightarrow (1)$: We come back to the semantics: $\interp{\input{vinj0.tikz}}_s= \left(\ket{x}\ket{y}\mapsto\delta_{Ax,Ay}\ket{Ax},1_n + 1_n,1_m\right)$ and $\interp{\input{vinj1.tikz}}_s= \left(\ket{x}\ket{y}\mapsto\delta_{x,y}\ket{Ax},1_n + 1_n,1_m\right)$. So for all $x$ and $y$ $\delta_{Ax,Ay}=\delta_{x,y}$, in other words $A$ is injective.\\
\end{proof}

\begin{proof}[\textbf{Proof of Lemma \ref{lm:rep}}]\phantomsection\label{pr:rep}
	By induction on the size $m\times n$ of $A$.
	If $n=m=1$, the result is exactly the zero and one axioms.
	if $m>1$:
	\begin{center}
		\setlength\tabcolsep{1pt}
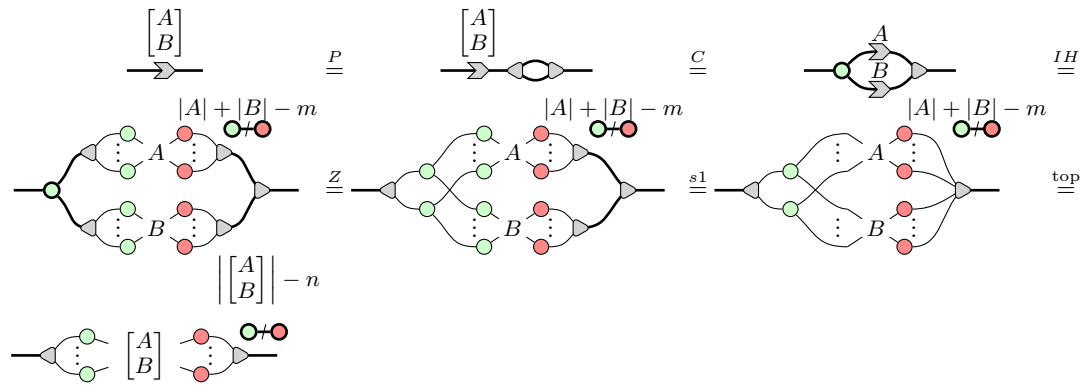
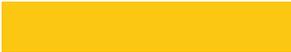
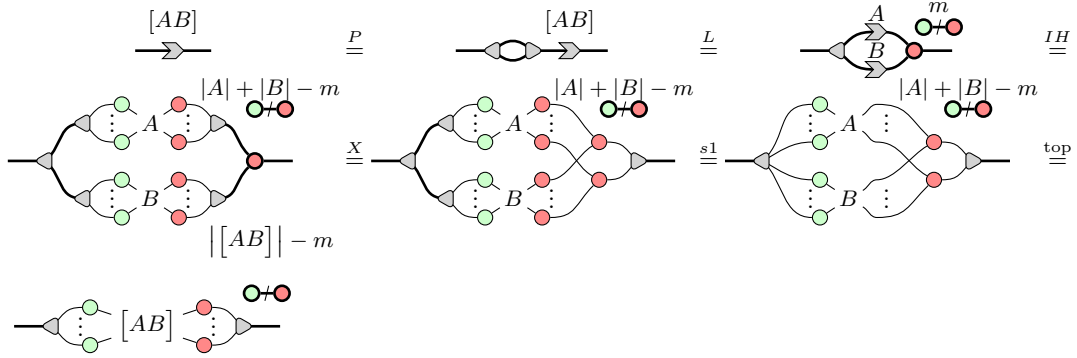
		\begin{tabular}{cccccc}
			\input{repn0.tikz}&$\stackrel{\hyperref[P]{P}}{=}$&\input{repn1.tikz}&$\stackrel{\hyperref[C]{C}}{=}$&\input{repn2.tikz}&$\stackrel{IH}{=}$\\
			\input{repn3.tikz}&$\stackrel{\hyperref[Z]{Z}}{=}$&\input{repn4.tikz}&$\stackrel{\hyperref[s1]{s1}}{=}$&\input{repn5.tikz}&$\stackrel{\hyperref[top]{\textup{top}}}{=}$\\
			\input{repn6.tikz}&&&&&
		\end{tabular}
	\end{center}
	
	if $n>1$:
	\begin{center}
		\setlength\tabcolsep{1pt}
		\begin{tabular}{cccccc}
			\input{repm0.tikz}&$\stackrel{\hyperref[P]{P}}{=}$&\input{repm1.tikz}&$\stackrel{\hyperref[L]{L}}{=}$&\input{repm2.tikz}&$\stackrel{IH}{=}$\\
			\input{repm3.tikz}&$\stackrel{\hyperref[X]{X}}{=}$&\input{repm4.tikz}&$\stackrel{\hyperref[s1]{s1}}{=}$&\input{repm5.tikz}&$\stackrel{\hyperref[top]{\textup{top}}}{=}$\\
			\input{repm6.tikz}&&&&&
		\end{tabular}
	\end{center}
	
\end{proof}

\begin{proof}[\textbf{Proof of Lemma \ref{lm:piv}}]\phantomsection\label{pr:piv}
	Up to permutation we take $u$ to be the first row in $ \Gamma$ and $v$ to be the first column. Then we have $\Gamma=\begin{bmatrix}
	1& A\\ B&C
	\end{bmatrix}$ where $A:n-1\to 1$, $B:1\to m-1$ and $C:n-1\to m-1$. The complete bipartite graph on $N_u \cap N_v$ has for biadjacency matrix $BA$ so the biadjacency matrix of $G_{n,m}\wedge uv$ is $\Gamma'=\begin{bmatrix}
	1& A\\ B&C{+}BA
	\end{bmatrix}$.
	
	\begin{center}
		\setlength\tabcolsep{1pt}
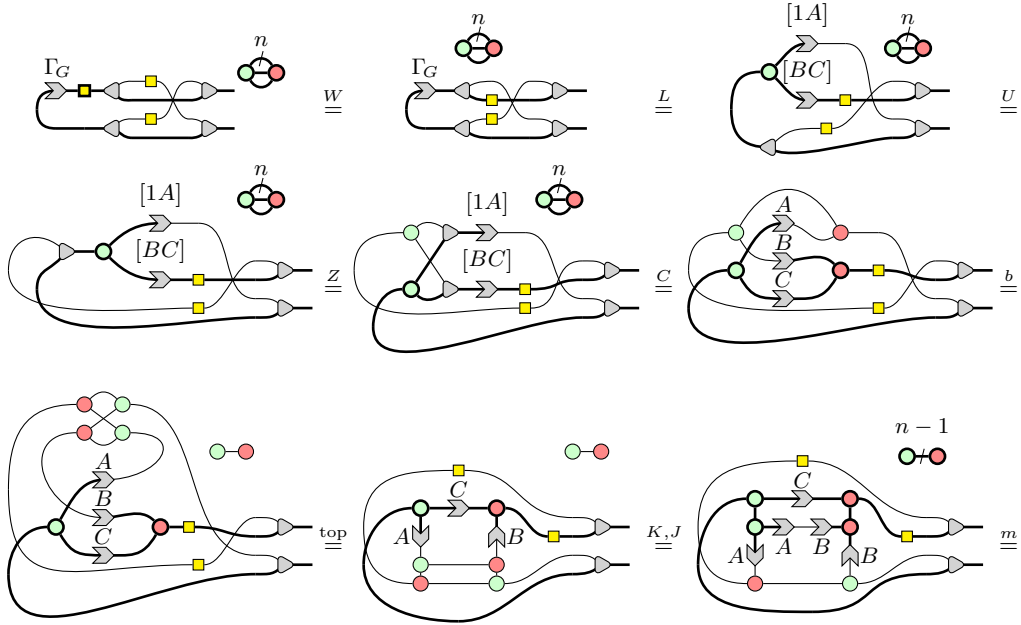
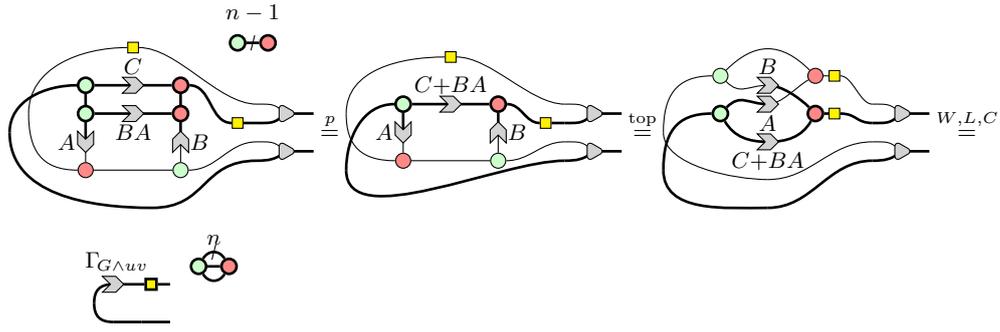
		\begin{tabular}{cccccc}
			\input{pivproof0.tikz}&$\stackrel{\hyperref[W]{W}}{=}$&\input{pivproof3.tikz}&$\stackrel{\hyperref[L]{L}}{=}$&\input{pivproof4.tikz}&$\stackrel{\hyperref[U]{U}}{=}$\\
			\input{pivproof5.tikz}&$\stackrel{\hyperref[Z]{Z}}{=}$&\input{pivproof6.tikz}&$\stackrel{\hyperref[C]{C}}{=}$&\input{pivproof7.tikz}&$\stackrel{\hyperref[b]{b}}{=}$\\
			\input{pivproof8.tikz}&$\stackrel{\hyperref[top]{\textup{top}}}{=}$&\input{pivproof9.tikz}&$\stackrel{\hyperref[K]{K},\hyperref[J]{J}}{=}$&\input{pivproof10.tikz}&$\stackrel{\hyperref[m]{m}}{=}$\\
			\end{tabular}
			
		\begin{tabular}{cccccc}
					\input{pivproof11.tikz}&$\stackrel{\hyperref[p]{p}}{=}$&\input{pivproof12.tikz}&$\stackrel{\hyperref[top]{\textup{top}}}{=}$&\input{pivproof13.tikz}&$\stackrel{\hyperref[W]{W},\hyperref[L]{L},\hyperref[C]{C}}{=}$\\
			\input{pivot0.tikz}&&&&&\\
		\end{tabular}
	\end{center}
	
\end{proof}

\begin{proof}[\textbf{Proof of Lemma \ref{lm:grphs}}]\phantomsection\label{pr:grphs}
	The partition component is straightforward typing. We use the characterization of $|G\rangle$ by its stabilizer \cite{hein2006entanglement}. We also use restrictions to subsets of $V(G)$, for example $\mathbf{v\cap N_u}$ is a vector of size $|N_u|$ and is zero if $v\notin N_u$.
	
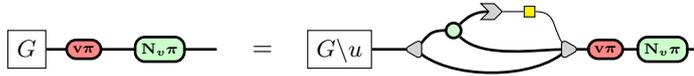
	\begin{equation*}
	\input{gsproofa0b.tikz}\quad=\quad\input{gsproofa1b.tikz}
	\end{equation*}
	
	We distinguish three cases. First if $u=v$:
	
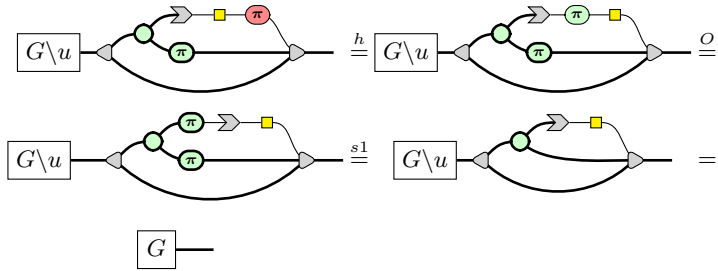
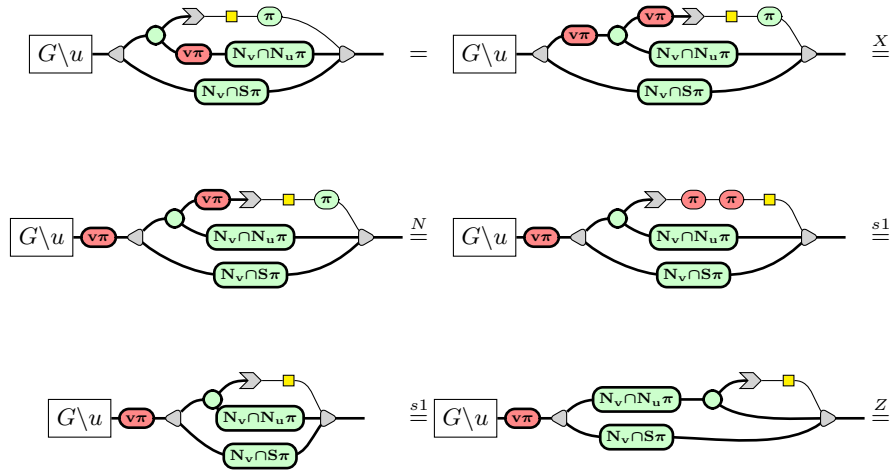
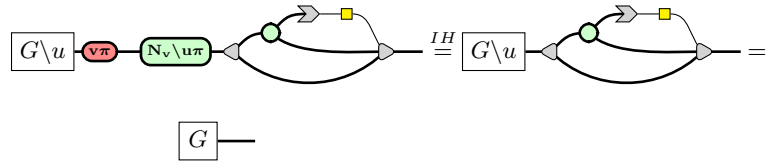
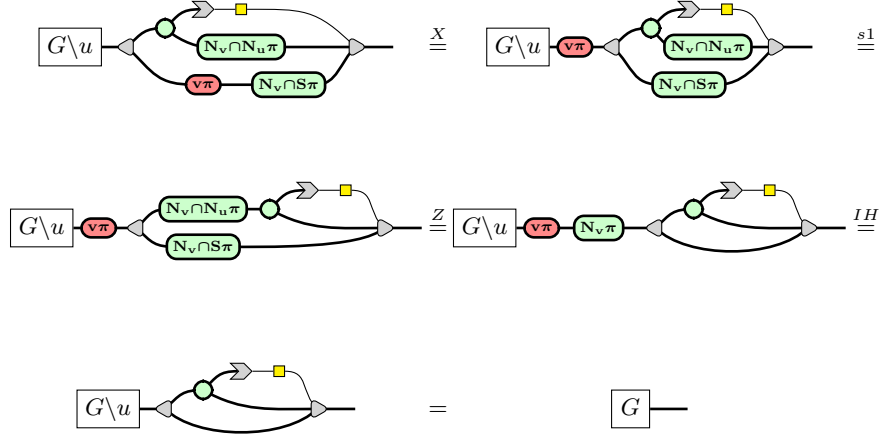
	\begin{center}
		\setlength\tabcolsep{1pt}
		\begin{tabular}{cccc}
			\input{gsproofa2bp.tikz}&$\stackrel{\hyperref[h]{h}}{=}$&\input{gsproofa3bp.tikz}&$\stackrel{\hyperref[O]{O}}{=}$\\[0.5cm]
			\input{gsproofa4bp.tikz}&$\stackrel{\hyperref[s1]{s1}}{=}$&\input{gsproofa5bp.tikz}&$=$\\[0.5cm]
			\begin{tikzpicture}
	\begin{pgfonlayer}{nodelayer}
		\node [style=box] (4) at (0, 0) {$G$};
		\node [style=none] (10) at (1.5, 0) {};
	\end{pgfonlayer}
	\begin{pgfonlayer}{edgelayer}
		\draw [style=very thick, in=180, out=0] (4) to (10.center);
	\end{pgfonlayer}
\end{tikzpicture}
&&&
		\end{tabular}
	\end{center}
	
	Then if $v\in N_u$:
	
	\begin{center}
		\setlength\tabcolsep{1pt}
		\begin{tabular}{cccc}
			\input{gsproofa2bpp.tikz}&$=$&\input{gsproofa3bpp.tikz}&$\stackrel{\hyperref[X]{X}}{=}$\\[1.5cm]
			\input{gsproofa4bpp.tikz}&$\stackrel{\hyperref[N]{N}}{=}$&\input{gsproofa5bpp.tikz}&$\stackrel{\hyperref[s1]{s1}}{=}$\\[1.5cm]
			\input{gsproofa6bpp.tikz}&$\stackrel{\hyperref[s1]{s1}}{=}$&\input{gsproofa7bpp.tikz}&$\stackrel{\hyperref[Z]{Z}}{=}$\\[1.5cm]
			\end{tabular}
				\begin{tabular}{cccc}
			\input{gsproofa8bpp.tikz}&$\stackrel{IH}{=}$&\input{gsproofa9bpp.tikz}&$=$\\[0.5cm]
			\begin{tikzpicture}
	\begin{pgfonlayer}{nodelayer}
		\node [style=box] (4) at (0, 0) {$G$};
		\node [style=none] (10) at (1.5, 0) {};
	\end{pgfonlayer}
	\begin{pgfonlayer}{edgelayer}
		\draw [style=very thick] (4) to (10.center);
	\end{pgfonlayer}
\end{tikzpicture}
&&&\\
		\end{tabular}
	\end{center}
	
	Finally if $v\in S_u$:
	
	\begin{center}
		\setlength\tabcolsep{1pt}
		\begin{tabular}{cccc}
			\input{gsproofa2b.tikz}&$\stackrel{\hyperref[X]{X}}{=}$&\input{gsproofa3b.tikz}&$\stackrel{\hyperref[s1]{s1}}{=}$\\[1.5cm]
			\input{gsproofa4b.tikz}&$\stackrel{\hyperref[Z]{Z}}{=}$&\input{gsproofa5b.tikz}&$\stackrel{IH}{=}$\\[1.5cm]
			\input{gsproofa6b.tikz}&$=$&\begin{tikzpicture}
	\begin{pgfonlayer}{nodelayer}
		\node [style=box] (4) at (0, 0) {$G$};
		\node [style=none] (10) at (1.5, 0) {};
	\end{pgfonlayer}
	\begin{pgfonlayer}{edgelayer}
		\draw [style=very thick] (4) to (10.center);
	\end{pgfonlayer}
\end{tikzpicture}
&
		\end{tabular}
	\end{center}

	It remains to take care of the scalar. We show by induction that $\langle G | |0\rangle^{\otimes n}=\frac{1}{\sqrt{2}^n}$, the case $n=1$ is true.
	
	\begin{center}
		\setlength\tabcolsep{1pt}
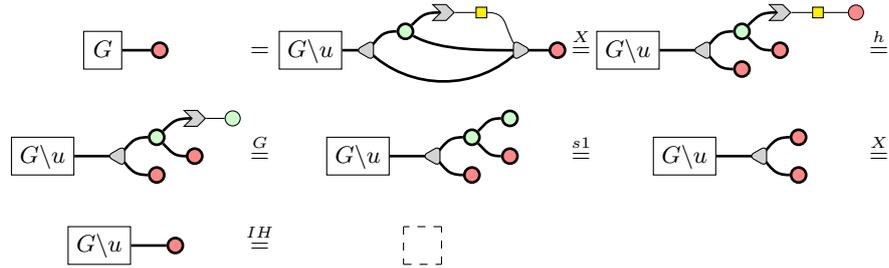
		\begin{tabular}{cccccc}
			\begin{tikzpicture}
	\begin{pgfonlayer}{nodelayer}
		\node [style=box] (0) at (0, 0) {$G$};
		\node [style=gred] (1) at (1.5, 0) {};
	\end{pgfonlayer}
	\begin{pgfonlayer}{edgelayer}
		\draw [style=very thick] (1) to (0);
	\end{pgfonlayer}
\end{tikzpicture}
&$=$&\input{gsproofb1.tikz}&$\stackrel{\hyperref[X]{X}}{=}$&\input{gsproofb2.tikz}&$\stackrel{\hyperref[h]{h}}{=}$\\[0.5cm]
			\input{gsproofb3.tikz}&$\stackrel{\hyperref[G]{G}}{=}$&\input{gsproofb4.tikz}&$\stackrel{\hyperref[s1]{s1}}{=}$&\input{gsproofb5.tikz}&$\stackrel{\hyperref[X]{X}}{=}$\\[0.5cm]
			\begin{tikzpicture}
	\begin{pgfonlayer}{nodelayer}
		\node [style=box] (0) at (0, 0) {$G\backslash u$};
		\node [style=gred] (2) at (2, 0) {};
	\end{pgfonlayer}
	\begin{pgfonlayer}{edgelayer}
		\draw [style=very thick] (2) to (0);
	\end{pgfonlayer}
\end{tikzpicture}
&$\stackrel{IH}{=}$&\input{gsproofb7.tikz}&&&
		\end{tabular}
	\end{center}
	
\end{proof}

\begin{proof}[\textbf{Proof of Theorem \ref{thm:locomp}}]\phantomsection\label{pr:locomp}
	
	The proof of this uses the triangle lemma, which follows from the rules of ZX calculus (a proof can be found in \cite{duncan2009graph}).
	
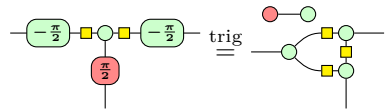
	\begin{lemma}\label{lm:trig}\label{trig}
		\begin{equation*}
		\input{trig0.tikz}\stackrel{\textup{trig}}{=}\input{trig1.tikz}
		\end{equation*}
	\end{lemma}
	
	We proceed by induction. Let $v\in N_u$.
	
	To avoid the permutations we choose a precise order on the vertices of $G$: $v,u,N_{u,v},N_{v,\lnot u},N_{u,\lnot v}, S_{u,v}$. Thus the we work with $\tau=id$ in the inductive definition of graph boxes.
	
	We want to compute:
	
	\begin{equation*}
	\input{lcproof0.tikz}\quad=\quad\input{lcproof1.tikz}
	\end{equation*}
	
	First we modify the middle gadget to clearly separate $u$ from the other neighbours of $v$:
	
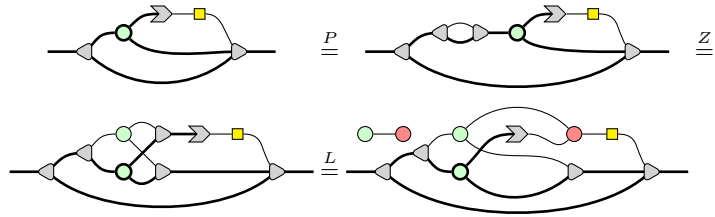
	\begin{center}
		\setlength\tabcolsep{1pt}
		\begin{tabular}{cccc}
			\input{alcproof1.tikz}&$\stackrel{\hyperref[P]{P}}{=}$&\input{alcproof2.tikz}&$\stackrel{\hyperref[Z]{Z}}{=}$\\
			\input{alcproof3.tikz}&$\stackrel{\hyperref[L]{L}}{=}$&\input{alcproof4.tikz}&
		\end{tabular}
	\end{center}
	
	Then we let the phases go through the modified gadget:
	
	\begin{center}
		\setlength\tabcolsep{1pt}
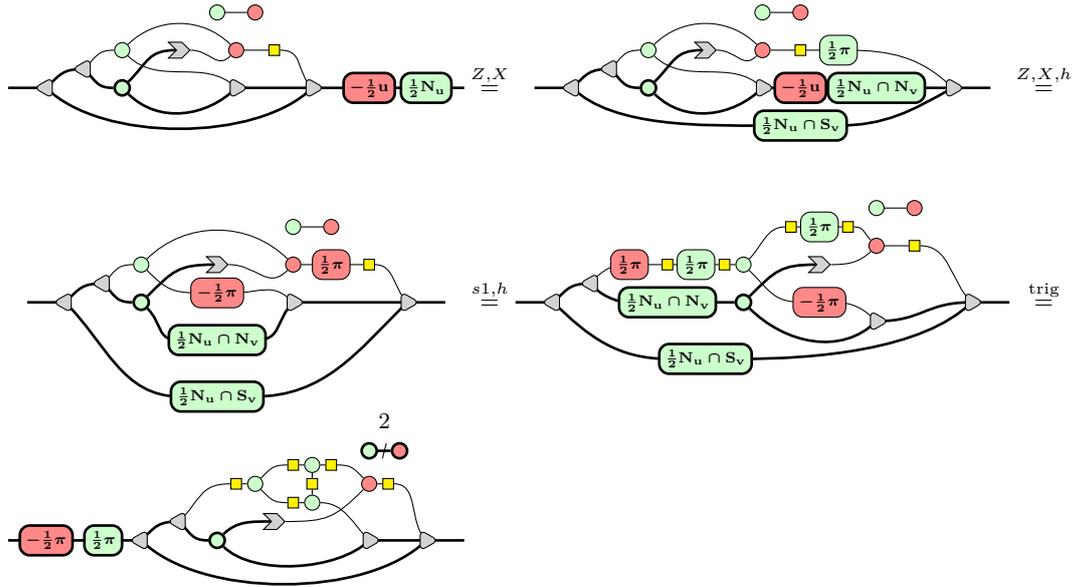
		\begin{tabular}{cccc}
			\input{blcproof4.tikz}&$\stackrel{\hyperref[Z]{Z},\hyperref[X]{X}}{=}$&\input{blcproof5.tikz}&$\stackrel{\hyperref[Z]{Z},\hyperref[X]{X},\hyperref[h]{h}}{=}$\\[1.2cm]
			\input{blcproof6.tikz}&$\stackrel{\hyperref[s1]{s1},\hyperref[h]{h}}{=}$&\input{blcproof7.tikz}&$\stackrel{\hyperref[trig]{\textup{trig}}}{=}$\\[1.2cm]
			\input{blcproof8.tikz}&&&
		\end{tabular}
	\end{center}
	
	The last step uses the triangle lemma \ref{lm:trig}. Plugging $G\backslash v$ we end up with:
	
	\begin{equation*}
	\input{lcproof8.tikz}
	\end{equation*}
	
Then simplifying the gadget, applying the induction hypothesis and the fact that $\left(G\backslash v\right)\star u=\left(G\star u\right)\backslash v$, we get:
	
	\begin{equation*}
	\input{lcproof0.tikz}\quad\stackrel{IH,\hyperref[h]{h},\hyperref[s1]{s1}}{=}\quad\input{lcproof9.tikz}
	\end{equation*}
	
	We modify the gadget in a way that the five sets, $v,N_{u,v},N_{v,\lnot u},N_{u,\lnot v}$ and $S_{u,v}$, are clearly separated:
	
	\begin{center}
		\setlength\tabcolsep{1pt}
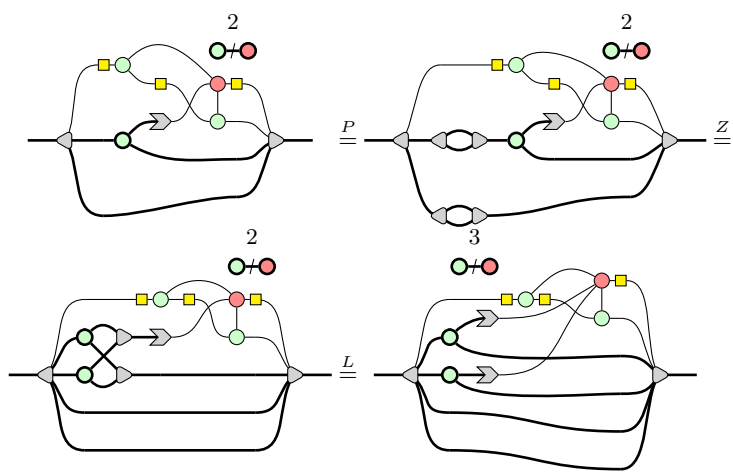
		\begin{tabular}{cccc}
			\input{clcproof9.tikz}&$\stackrel{\hyperref[P]{P}}{=}$&\input{clcproof10.tikz}&$\stackrel{\hyperref[Z]{Z}}{=}$\\
			\input{clcproof11.tikz}&$\stackrel{\hyperref[L]{L}}{=}$&\input{clcproof12.tikz}&
		\end{tabular}
	\end{center}
	
	Now we want to extract $u$ from $\left(G\star u\right)\backslash v$. To do this we use the usual gadget and the fact that ${N_u}'=N_{u,v}\cup N_{u,\lnot v}$ and ${S_u}'=N_{\lnot u,v}\cup S_{u,v}$. We compose it with the previous gadget:
	
	\begin{equation*}
	\input{dlcproof13.tikz}
	\end{equation*}
	
	Now we compute:
	
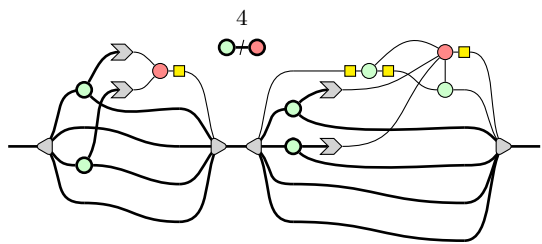
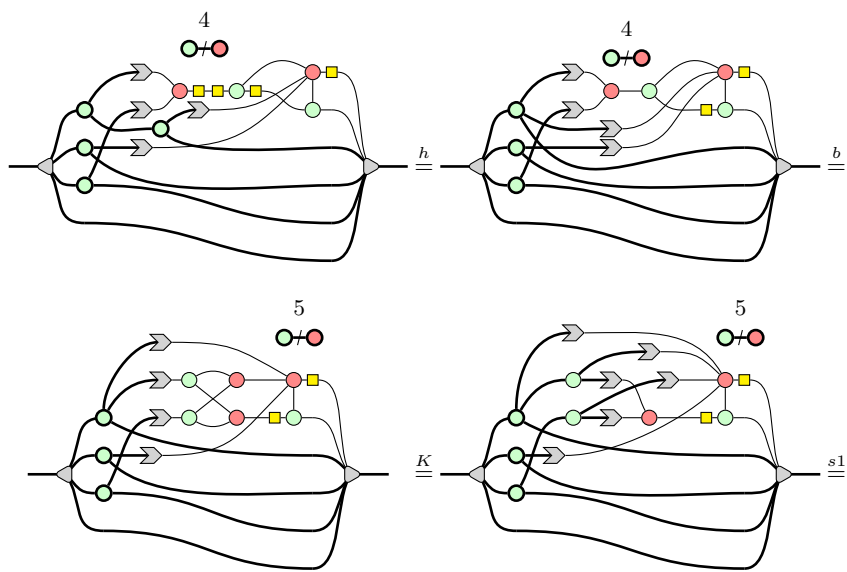
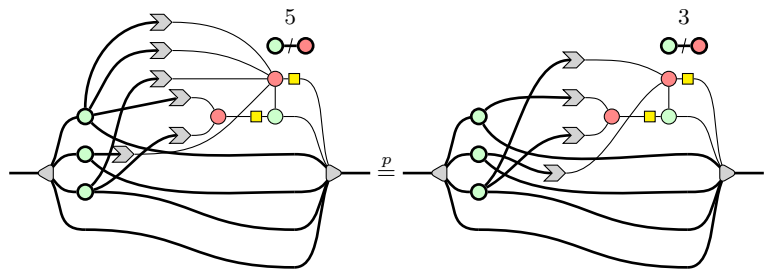
	\begin{center}
		\setlength\tabcolsep{1pt}
		\begin{tabular}{cccc}
			\input{dlcproof14.tikz}&$\stackrel{\hyperref[h]{h}}{=}$&\input{dlcproof15.tikz}&$\stackrel{\hyperref[b]{b}}{=}$\\[1.5cm]
			\input{dlcproof16.tikz}&$\stackrel{\hyperref[K]{K}}{=}$&\input{dlcproof17.tikz}&$\stackrel{\hyperref[s1]{s1}}{=}$\\[1.5cm]
				\end{tabular}
				\begin{tabular}{cccc}
			\input{dlcproof18.tikz}&$\stackrel{\hyperref[p]{p}}{=}$&\input{dlcproof19.tikz}&
		\end{tabular}
	\end{center}
	
	Now we reform a gadget on the left to reincoporate $u$ in $\left(G\star u\right)\backslash v \backslash u$.
	
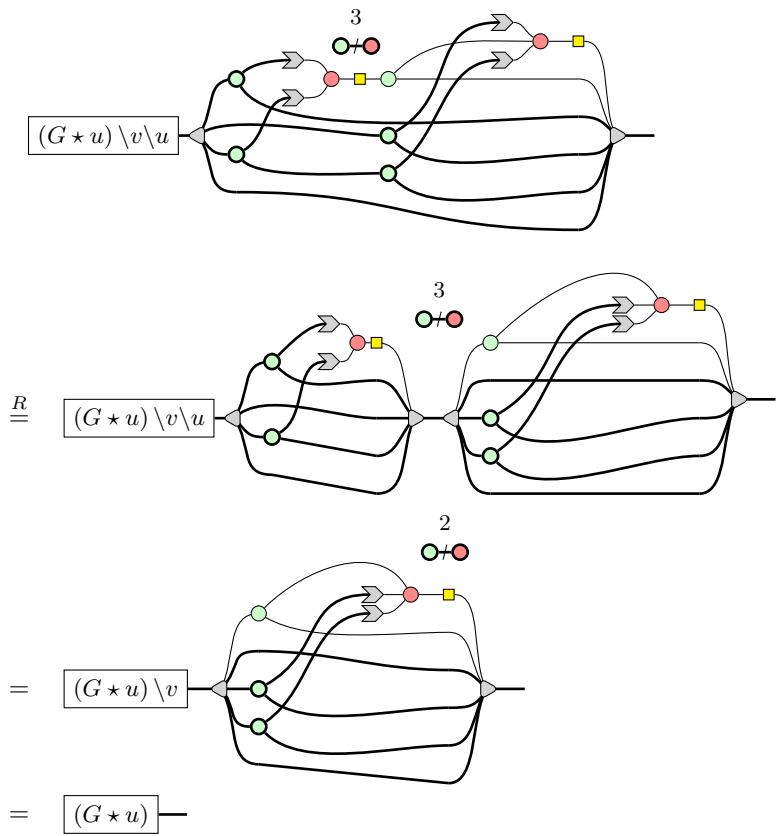
	\begin{align*}
	&\quad\input{elcproof20.tikz}\\
	&\stackrel{\hyperref[R]{R}}{=}\quad\input{elcproof21.tikz}\\
	&=\quad\input{elcproof22.tikz}\\
	&=\quad\begin{tikzpicture}
	\begin{pgfonlayer}{nodelayer}
		\node [style=box] (1) at (0, 0) {$\left(G\star u\right) $};
		\node [style=none] (155) at (2, 0) {};
	\end{pgfonlayer}
	\begin{pgfonlayer}{edgelayer}
		\draw [style=very thick] (1) to (155.center);
	\end{pgfonlayer}
\end{tikzpicture}

	\end{align*}
	
	The two last steps are done by recognizing the corresponding gadgets up to permutation. 
\end{proof}

\begin{proof}[\textbf{Proof of Lemma \ref{lm:iso}}]\phantomsection\label{pr:iso}
	We compute the composition of the encoder and the decoder:
	
	\begin{center}
		\setlength\tabcolsep{1pt}
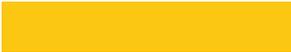
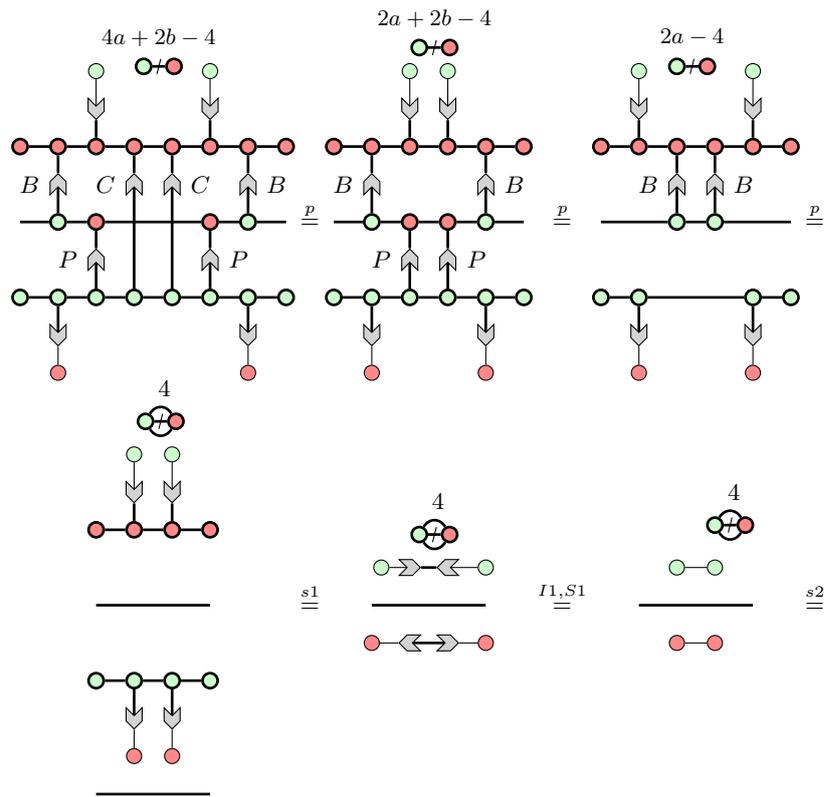
		\begin{tabular}{cccccc}
			\input{cpcdproof0.tikz}&$\stackrel{\hyperref[p]{p}}{=}$&\input{cpcdproof1.tikz}&$\stackrel{\hyperref[p]{p}}{=}$&\input{cpcdproof2.tikz}&$\stackrel{\hyperref[p]{p}}{=}$\\
			\input{cpcdproof3.tikz}&$\stackrel{\hyperref[s1]{s1}}{=}$&\input{cpcdproof4.tikz}&$\stackrel{\hyperref[I1]{I1},\hyperref[S1]{S1}}{=}$&\input{cpcdproof5.tikz}&$\stackrel{\hyperref[s2]{s2}}{=}$\\
			\begin{tikzpicture}
	\begin{pgfonlayer}{nodelayer}
		\node [style=none] (79) at (0, 0) {};
		\node [style=none] (103) at (3, 0) {};
	\end{pgfonlayer}
	\begin{pgfonlayer}{edgelayer}
		\draw [style=very thick] (79.center) to (103.center);
	\end{pgfonlayer}
\end{tikzpicture}
&&&&&
		\end{tabular}
	\end{center}
	
\end{proof}

\begin{proof}[\textbf{Proof of Lemma \ref{lm:err}}]\phantomsection\label{pr:err}$\quad$\\
	
	\begin{center}
		\setlength\tabcolsep{1pt}
		\begin{tabular}{cccccc}
			\input{verproof0.tikz}&$\stackrel{\hyperref[s1]{s1}}{=}$&\input{verproof1.tikz}&$\stackrel{\hyperref[N]{N},\hyperref[s1]{s1}}{=}$&\input{verproof2.tikz}&$\stackrel{\hyperref[N]{N}}{=}$\\
			\input{verproof3.tikz}&$\stackrel{\hyperref[s1]{s1}}{=}$&\input{verproof4.tikz}&&&\\
		\end{tabular}
	\end{center}
	
	\begin{center}
		\setlength\tabcolsep{1pt}
		\begin{tabular}{cccccc}
			\input{rerproof0.tikz}&$\stackrel{\hyperref[s1]{s1}}{=}$&\input{rerproof1.tikz}&$\stackrel{\hyperref[O]{O},\hyperref[s1]{s1}}{=}$&\input{rerproof2.tikz}&$\stackrel{\hyperref[O]{O}}{=}$\\
			\input{rerproof3.tikz}&$\stackrel{\hyperref[s1]{s1}}{=}$&\input{rerproof4.tikz}&&&\\
		\end{tabular}
	\end{center}
	
\end{proof}

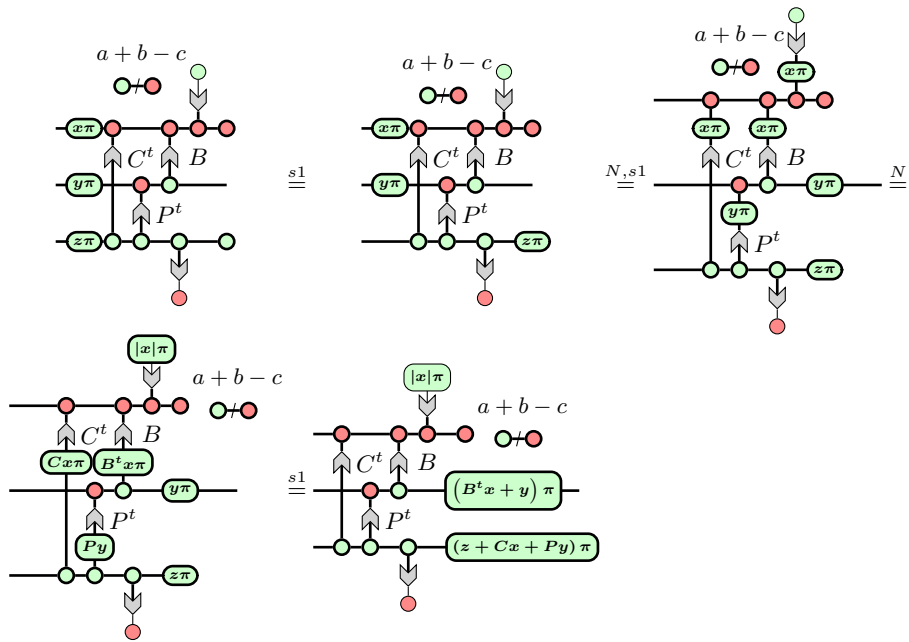
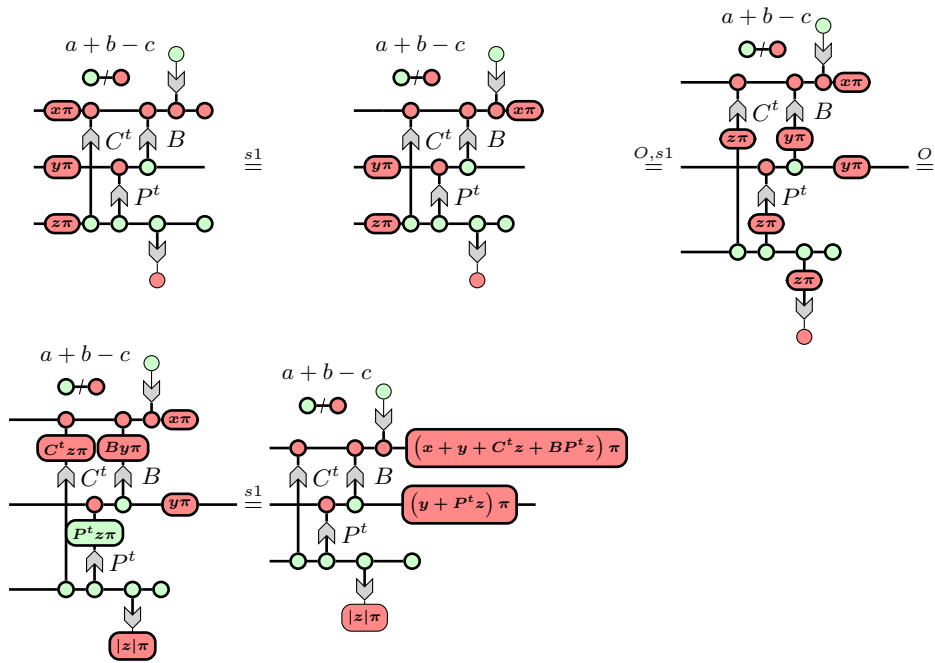
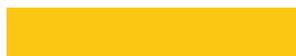
\end{document}